\documentclass[runningheads]{llncs}

\usepackage[mobile]{eccv}

\usepackage{graphicx}
\usepackage[accsupp]{axessibility}
\usepackage{amsmath}
\usepackage{amssymb}
\usepackage{booktabs}
\usepackage{blindtext}
\usepackage{colortbl}
\usepackage{hyperref}
\usepackage{lipsum}
\usepackage{microtype}
\usepackage{overpic}
\usepackage{pifont}
\usepackage{tikz}
\usepackage{xspace}
\usepackage{xcolor}
\usetikzlibrary{shadows.blur}
\usepackage{rotating}
\usepackage[ruled,vlined,linesnumbered]{algorithm2e}

\definecolor{turquoise}{cmyk}{0.65,0,0.1,0.3}
\definecolor{purple}{rgb}{0.65,0,0.65}
\definecolor{dark_green}{rgb}{0, 0.5, 0}
\definecolor{orange}{rgb}{0.8, 0.6, 0.2}
\definecolor{red}{rgb}{0.8, 0.2, 0.2}
\definecolor{darkred}{rgb}{0.6, 0.1, 0.05}
\definecolor{blueish}{rgb}{0.0, 0.3, .6}
\definecolor{light_gray}{rgb}{0.7, 0.7, .7}
\definecolor{pink}{rgb}{0.9, 0, 0.6}
\definecolor{greyblue}{rgb}{0.25, 0.25, 1}
\definecolor{teal}{rgb}{0.0, 0.4, 0.4}
\definecolor{chocolate}{rgb}{1.0, 0.4, 0.0}

\definecolor{figred}{rgb}{0.8352941176470589, 0.24313725490196078, 0.30980392156862746}
\definecolor{figgreen}{rgb}{0.3070588235294118, 0.5498039215686275, 0.2749019607843137}
\definecolor{figblue}{rgb}{0.19607843137254902, 0.5333333333333333, 0.7411764705882353}

\renewcommand{\paragraph}[1]{\vspace{.5em}\noindent\textbf{#1.}}

\DeclareMathOperator*{\argmin}{arg\,min}

\newcommand{\losst}[1]{\mathcal{L}_\text{#1}}

\newcommand{\real}{\mathbb{R}}
\newcommand{\trainrays}{\mathcal{R}}
\newcommand{\ray}{\mathbf{r}}

\newcommand{\raydirection}{\mathbf{d}}
\newcommand{\radiance}{\mathbf{c}}

\newcommand{\opacity}{\alpha}
\newcommand{\transmittance}{T}

\newcommand{\primal}{\mathbf{p}}
\newcommand{\normal}{\mathbf{n}}

\newcommand{\radius}{r}
\newcommand{\cell}{\mathbf{V}}

\newcommand{\power}{\mathbf{P}}
\newcommand{\powers}{\mathcal{P}}

\newcommand{\displacement}{\mathbf{d}}

\newcommand{\SupplementaryMaterial}{\textcolor{purple}{supplementary material}\xspace}
\newcommand{\Webpage}{\textcolor{purple}{webpage}\xspace}

\newcommand{\mipnerf}{MipNeRF~360~\cite{mipnerf360}\xspace}
\newcommand{\dlthreedv}{DL3DV~\cite{dl3dv}\xspace}

\usepackage{soul}
\setuldepth{foobar}

\newcommand{\cmark}{\textcolor{green}{\ding{51}}} 
\newcommand{\xmark}{\textcolor{red}{\ding{55}}}   

\usepackage[shortlabels,inline]{enumitem}
\setlist[itemize]{noitemsep,leftmargin=*,topsep=0em}
\setlist[enumerate]{noitemsep,leftmargin=*,topsep=0em}

\usepackage{eccvabbrv}

\usepackage{orcidlink}

\begin{document}

\title{Power Foam: Unifying Real-Time Differentiable Ray Tracing and Rasterization} 

\titlerunning{Power Foam}

\author{
Shrisudhan Govindarajan\thanks{Equal contribution}$^{,\dagger}$\inst{1,4}\orcidlink{0000-0002-3546-8223} \and
Daniel Rebain\protect\footnotemark[1]\inst{2}\orcidlink{0000-0003-4691-7909} \and
Dor Verbin\inst{3}\orcidlink{0000-0001-8798-3270} \and
Kwang Moo Yi\inst{2}\orcidlink{0000-0001-9036-3822} \and
Anish Prabhu\inst{4}\orcidlink{0009-0007-2813-8457} \and
Andrea Tagliasacchi\inst{1,5}\orcidlink{0000-0002-2209-7187}
}

\authorrunning{S.~Govindarajan et al.}

\institute{
Simon Fraser University \quad 
University of British Columbia \quad
Google Deepmind \and
Google \quad
University of Toronto
}

\maketitle
\begingroup
\renewcommand\thefootnote{$\dagger$}
\footnotetext{Work done at Google}
\endgroup

\begin{center}
    Project page: \url{https://powerfoam.github.io}
\end{center}

\vspace{-3em}
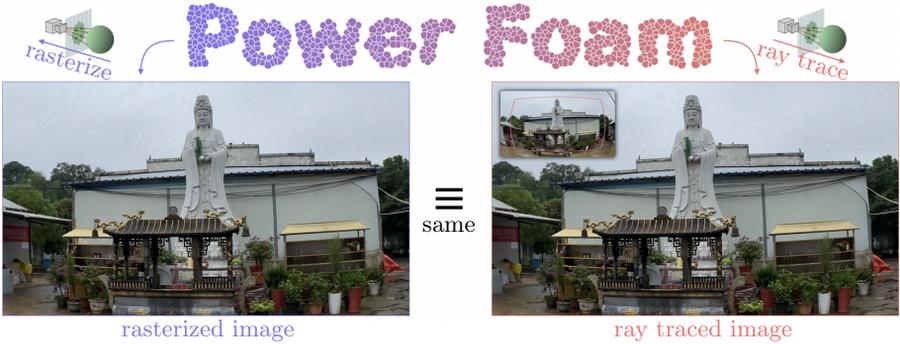
\begin{figure*}[h]
    \centering
    \begin{overpic}[width=\linewidth]{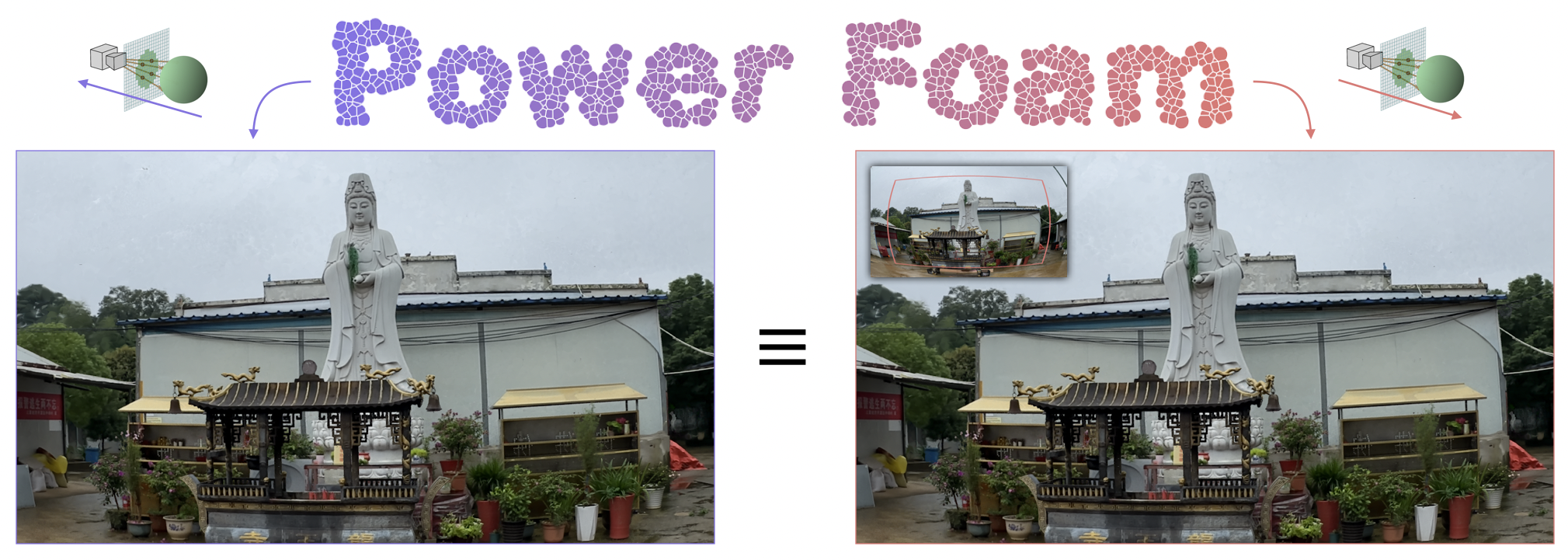}
    \put(47, 10){same}
    \put(14,-1.25){\textcolor[rgb]{0.5071, 0.45, 0.9}{rasterized image}}
    \put(68, -1.25){\textcolor[rgb]{0.9, 0.45, 0.45}{ray traced image}}
    \begin{turn}{344}
        \put(-4.75, 28.75){\textcolor[rgb]{0.5071, 0.45, 0.9}{rasterize}}
        \put(72.25, 51){\textcolor[rgb]{0.9, 0.45, 0.45}{ray trace}}
    \end{turn}
    \end{overpic}
    \caption{
    {\bf Teaser -- }
    we introduce a differentiable 3D representation that unifies the flexibility of foam-based ray tracing with the efficiency of modern rasterization pipelines.
    In the center, we illustrate the 2D structure of our bounded power diagram model.
    By utilizing the bounded power diagram with controllable cell extents, our method generates spherically bounded primitives that are highly amenable to tile-based culling, enabling efficient rasterization (left). 
    Simultaneously, our representation preserves the constant-time ray traversal characteristics through the volumetric mesh.
    Our ray traced results, e.g. the fisheye image in the inset (right), are identical to the result of distorting a pinhole render, whereas other methods lose some fidelity in approximation.
    Notably, Power Foam achieves consistent, high-frame-rate outputs under both rendering paradigms.
    }
    \label{fig:teaser}
\vspace{-3em}
\end{figure*}

\begin{abstract}
We introduce a differentiable 3D representation that unifies the ray tracing capabilities of foam-based ray tracing with the efficiency of modern rasterization pipelines. While prior foam representations enable constant-time ray traversal through an explicit volumetric partition of space, their potentially unbounded cells hinder efficient tile-based rasterization. We address this limitation by generalizing Voronoi foams to bounded power diagrams with controllable cell extents, enabling spatially bounded primitives without requiring expensive Delaunay triangulations during training. We further introduce an oriented surface formulation that explicitly models interfaces between interior and exterior regions, and decouple geometry from appearance by embedding differentiable texture directly on these surfaces. Together, these contributions yield a representation that preserves state-of-the-art ray tracing efficiency while achieving rasterization performance competitive with current generation 3DGS, providing a practical path toward unified real-time differentiable rendering.
\end{abstract}    
\section{Introduction}
\label{sec:intro}

Recent advances in differentiable rendering have led to highly optimized scene representations such as 3D Gaussian Splatting (3DGS)~\cite{3dgs}, capable of producing high-quality, photorealistic renderings at faster-than-real-time frame rates.
This success has been enabled by the computational efficiency of renderers based on rasterization, which can draw high-resolution frames with minimal computational cost and form the backbone of most real-time graphics systems, including modern video game engines, which rely on highly optimized rasterization pipelines for interactive performance.

While rasterization-based Gaussian Splatting renders entire images at once, researchers are increasingly exploring ray tracing formulations~\cite{3dgrt,radfoam}.
The ability to evaluate the color of each ray independently enables simulation of complex effects like reflection and refraction (such as those demonstrated by \cite{radfoam}), which are not compatible with rasterization, and can potentially enable more advanced techniques such as Monte Carlo path tracing.

The distinct advantages of these rendering paradigms have motivated the pursuit of a unified formulation.
Recent methods extend Gaussian Splatting to support efficient ray tracing~\cite{3dgrt} and mathematically align the rendering equations to allow hybrid pipelines, where the same representation can be rasterized for speed or traced for complex light transport~\cite{3dgut}.
However, because Gaussian primitives are unstructured and overlap heavily, they do not define a true volumetric partition of space.
Consequently, ray tracing these scenes necessitates the construction and traversal of a Bounding Volume Hierarchy (BVH).
While hardware acceleration has enabled real-time ray tracing using this technique, it fundamentally couples the complexity of ray traversal to the number of primitives in the scene.

Meanwhile, Radiant Foam~\cite{radfoam} models an explicit partition of space via a volumetric mesh, eliminating overlap and 
reducing per-primitive ray traversal from the logarithmic complexity of a BVH to constant time.
This representation thus unlocks highly efficient ray tracing, but would it be possible to achieve a truly unified formulation in which foams can \textit{also} be rasterized efficiently?
The central challenge is that volumetric cells in a foam can become very large or even unbounded, which interferes with the frustum culling and spatial locality assumptions that underpin efficient tile-based rasterization.

Our goal is therefore to design a representation that preserves the constant-time ray traversal complexity of foam-based models while recovering the spatial locality required for efficient rasterization. 
We achieve this through three key contributions.

\begin{itemize}[label=\textbullet]
    \item \textbf{Bounded power diagrams:} In Radiant Foam, space is partitioned using a Voronoi diagram.
Here, we generalize this construction to a restricted (i.e. bounded) power diagram, specifically, the structure dual to the weighted $\alpha$-complex~\cite{edelsbrunner2003shape}, where each site is parametrized with a controllable radius.
This additional degree of freedom allows us to explicitly regulate cell extent, preventing unbounded regions and producing spatially localized cells that are amenable to efficient tile-based rasterization, while avoiding the need to construct expensive full Delaunay triangulations during training.
    \item \textbf{Surface modelling:} In Radiant Foam, surfaces emerge implicitly as interfaces between adjacent high- and low-density cells.
In contrast, we introduce an oriented point representation that explicitly partitions each cell into interior and exterior regions, yielding a more direct representation of surfaces.
    \item \textbf{Decoupling geometry/appearance:} In Radiant Foam, representing highly textured regions often requires increasing the number of volumetric cells.
Instead, we model a differentiable texture directly on the surfaces induced by oriented points, disentangling geometry from appearance and significantly reducing the required cell budget.
\end{itemize}

Together, these contributions yield a representation that unifies the strengths of both paradigms: it retains the constant-time ray traversal and state-of-the-art efficiency of foam-based ray tracing, while introducing the spatial localization and surface structure necessary for rasterization performance competitive with current generation 3DGS.
\section{Related Work}
\label{sec:related}
While early coordinate-based approaches such as Scene Representation Networks (SRNs)~\cite{srn} and Neural Volumes~\cite{neuralvolumes} pioneered the use of differentiable rendering for neural scene modeling, Neural Radiance Fields (NeRF)~\cite{nerf} significantly advanced this paradigm by enabling high-fidelity volumetric reconstruction of complex, large-scale datasets.
Despite the exceptional photorealism achieved by NeRF, its reliance on dense MLP evaluations along each ray remains computationally prohibitive for real-time applications. 
Subsequent research has attempted to address these latency issues by incorporating hybrid structures such as voxel grids~\cite{plenoxels}, multi-resolution hash tables~\cite{instantngp} and specialized sampling schemes~\cite{mipnerf, mipnerf360}. 
Nevertheless, the requirement of ray marching remains a significant bottleneck for Monte Carlo ray tracing methods, which require the evaluation of numerous paths per pixel to resolve complex secondary effects and global illumination.

\paragraph{Particle-Based Rasterization}
The emergence of 3D Gaussian Splatting (3DGS)~\cite{3dgs} moved away from continuous volumes toward unstructured particle-based representations. 
Using a tile-based rasterization approach, 3DGS projects anisotropic Gaussians onto the image plane, enabling extremely high frame rates on modern GPUs. 
This efficiency stems from the ability to sort and alpha-blend primitives within localized tiles. 
Building on this success, several variations have been proposed to improve representation and optimization~\cite{3dgs-mcmc, betasplatting}. 
However, this rasterization-first design is inherently limited: it lacks the spatial connectivity required for efficient ray traversal, making it difficult to resolve effects like  shadows, reflections, or secondary bounces without significant architectural modifications.

\paragraph{Particle-Based Ray Tracing}
Ray tracing particles requires efficient intersection testing, which is challenging for unstructured collections of Gaussians. 
To address this, several works have proposed building hierarchical acceleration structures, such as Bounding Volume Hierarchies (BVH), over the splatting primitives~\cite{3dgrt, ever}. 
While these structures enable ray-primitive intersections, the heavily overlapping nature of Gaussians leads to high depth complexity and redundant computations during traversal.

Alternative structures like Radiant Foam~\cite{radfoam} attempt to solve this by using volumetric partitions. 
This method allows for ``constant-time'' neighbor-to-neighbor traversal, which is significantly faster than hierarchical search for continuous ray paths. 
However, these volumetric meshes often contain unbounded cells, which limits the ability to effectively perform frustum culling for efficient rasterization—a constraint we address through our \textit{localized} Power Diagram.

\paragraph{Unified Representations}
The pursuit of a unified representation aims to integrate the tile-based rasterization efficiency of 3DGS with the topological connectivity of volumetric meshes. 
Contemporary models such as 3DGUT~\cite{3dgut} have explored this unification using Gaussian primitives, but rely on BVH-based ray-traversal as they lack inherent primitive connectivity. 
Similarly, Radiance Meshes~\cite{radiancemeshes} utilize tetrahedral meshes that support constant-complexity ray traversal, yet the authors nonetheless resort to BVH-based ray tracing inference.
Consequently, both methods incur significant rendering speed overhead compared to the adjacency-based traversal demonstrated by Radiant Foam~(for example, a 2.5$\times$ speedup from 3DGRT to Radiant Foam).
Our proposed representation addresses this performance gap by providing the finite spatial bounds necessary for efficient rasterization while preserving the high-speed, neighbor-to-neighbor traversal benefits of volumetric foams.

\section{Method}
\label{sec:method}

\begin{figure}[t]
  \centering
  \begin{overpic}[width=\linewidth]{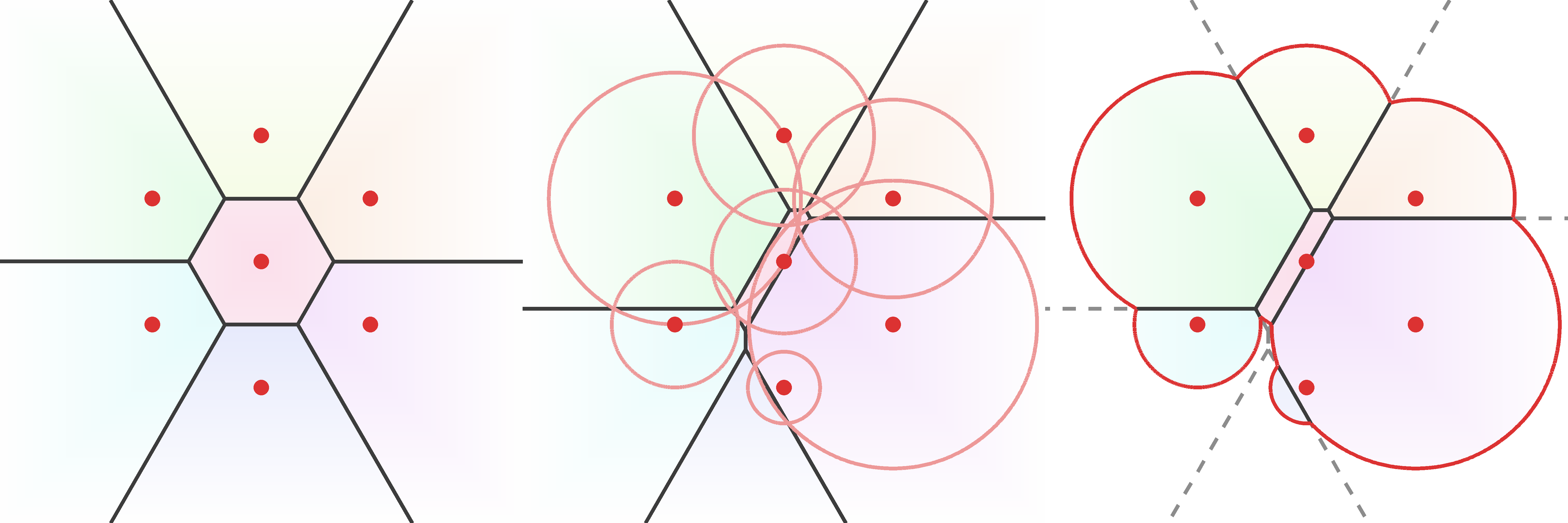}
  \end{overpic}
  \caption{\textbf{Comparison of volumetric mesh types} -- the Voronoi diagram (left) constructs cell faces from planes equidistant to the cell sites, while the power diagram (center) constructs them based on the radii associated with each cell. By using these radii to define bounding spheres for each cell (right), we can ensure that all parts of the cell boundaries will have gradients with respect to all cell parameters.}
  \label{fig:dual_meshes}
\end{figure}

\subsection{Preliminaries: Radiant Foam}
Our method extends Radiant Foam~\cite{radfoam}, so we begin with a brief review:
Radiant Foam models the 3D scene using a non-overlapping, differentiable volumetric mesh.
The volumetric mesh is constructed as a Voronoi diagram that partitions the scene into convex polyhedral cells $\cell_1, ..., \cell_N \subseteq \real^3$.
These cells are generated by a set of sites $\primal_1, ..., \primal_N \in \real^3$, which simultaneously serve as the vertices of the dual Delaunay triangulation.
Each cell $\cell_i$ comprises the region of 3D space closest to its respective site $\primal_i$:
\begin{align}
    \cell_i = \{\mathbf{x} \in \real^3 : \argmin_j ||\mathbf{x} - \primal_j|| = i\}.\label{eq:voronoi}
\end{align}
To support volume rendering, Radiant Foam equips each cell with a learnable volume density value and a set of RGB Spherical Harmonic (SH) coefficients used for modeling the view-dependent color of the cell.
The standard volume rendering integral can then be evaluated exactly as a sum over the segments of intersection between a ray and the Voronoi cells~\cite{digest}.

Ray tracing is highly efficient in this representation because Voronoi cells are \textit{convex} and share faces. 
The adjacency information provided by the dual Delaunay triangulation allows a ray to ``walk'' from one cell to the next by checking all faces, each corresponding to a Delaunay edge, to find where the ray crosses into the next cell.
Govindarajan \etal~\cite{radfoam} show that this process runs in amortized constant time per transition, as the expected number of neighbors for any cell depends only on the number of spatial dimensions, not the size of the mesh~\cite{meijering1953interface}.

\subsection{Ray Tracing and Rasterizing Bounded Power Diagrams}

Our goal is to construct a representation that can be rasterized as well as ray traced.
This will enable fast rendering using rasterization, while also supporting ray tracing for the flexibility it allows in modeling light transport phenomena like reflection and refraction.
While foam representations are natively amenable to ray tracing, efficient rasterization requires that primitives be bounded.
Ideally, we want bounds which are simple to compute in screen space -- such that we can easily determine which image tiles they intersect -- and which will exclude any pixels that the primitive will never contribute to.
An unbounded foam structure fails both of these tests; testing for intersection with image tiles requires an unwieldy computation of the hull of the projected convex in screen space, which may also include large areas where the cell is completely occluded.

The easiest solution to this issue is to restrict each Voronoi foam cell to its intersection with a rasterization-friendly bounding primitive such as a sphere.
Unfortunately, naïvely adding learnable bounding primitives to each cell for this purpose would create new problems, such as lacking gradients when the bound is much larger than the cell (or vice-versa), as well as faces being induced between cells despite their bounds not intersecting.
Thankfully, computational geometry has already devised a structure with ideal properties: the weighted $\alpha$-complex~\cite{edelsbrunner2003shape}, or more specifically, its dual, which we refer to as the \textit{bounded power diagram}; see~\cref{fig:dual_meshes}.

\begin{figure}[t]
  \centering
  \begin{overpic}[width=\linewidth]{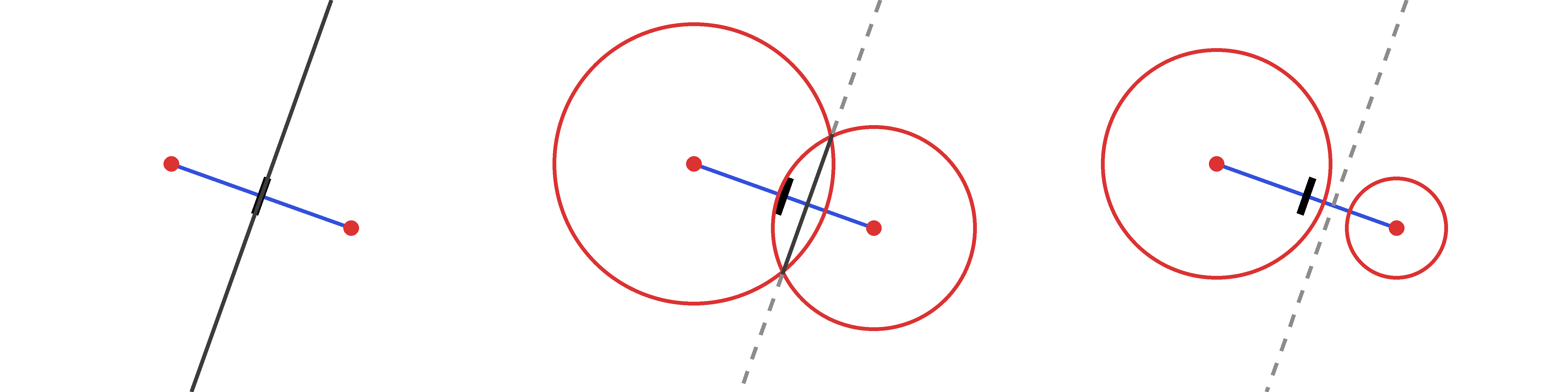}
  \end{overpic}
  \caption{\textbf{Power cell faces depend on radius} -- while the Voronoi diagram faces are always exactly equidistant between sites (left), the faces of the power cell are determined by both sphere centers and radii. Specifically, the power cell face between two neighboring cells lies on the \emph{radical plane} of the two spheres. For overlapping spheres, this plane contains the circle of intersection between them (middle), and for non-overlapping, it always lies \emph{outside} the spheres (right).}
  \label{fig:radical_vs_voronoi}
\end{figure}

The power diagram generalizes the Voronoi diagram by introducing a weighted distance, such that each cell $\power_i$ is parameterized by a primal sites $\primal_i$ and an associated squared radius (also known in the literature as a weight) $\radius_i^2$:
\begin{align}
    \power_i = \{\mathbf{x} \in \real^3 : \argmin_j ||\mathbf{x} - \primal_j||^2 - \radius_j^2 = i\}
\end{align}
This weighted distance is known as the \textit{power} of the point $\mathbf{x}$ with respect to the sphere defined by $\primal_i$ and $\radius_i$.
Making this radius learnable, and also taking it as the radius of a bounding sphere, solves both of the problems mentioned above.
As shown in~\cref{fig:radical_vs_voronoi}, the radius now affects both the cell-to-cell faces and spherical cell boundaries, so it always has gradients.
Also, unlike Voronoi cells, non-intersecting bounded power cells will never induce a face that would cause non-local interaction; see \cref{fig:nonlocal_faces}.

\begin{figure}[t]
  \centering
  \begin{minipage}[t]{0.48\linewidth}
    \begin{overpic}[width=\linewidth]{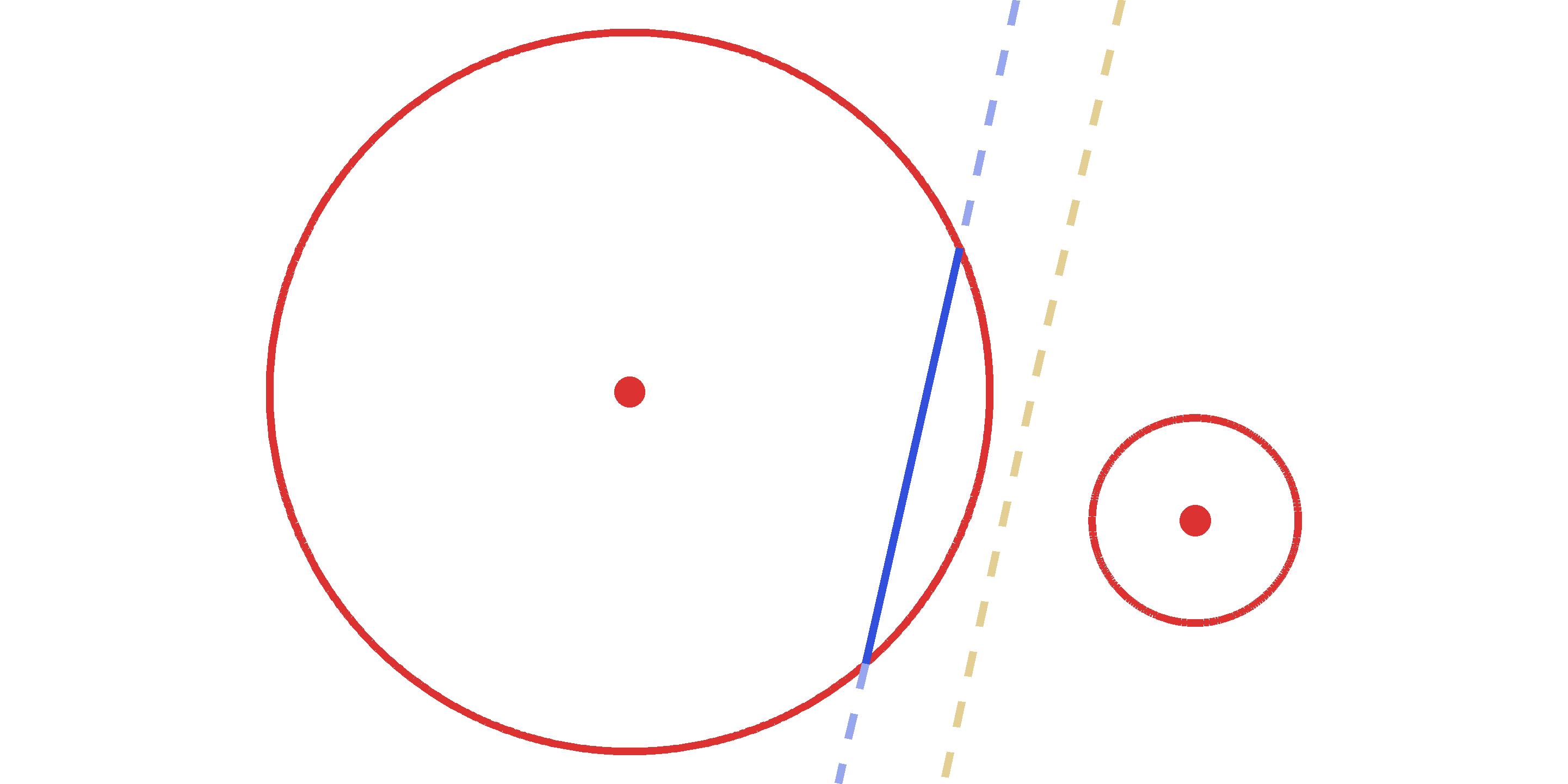}
    \end{overpic}
    \captionof{figure}{\textbf{Avoiding non-local faces with the radical plane} -- while it would be possible to construct bounded cells using the Voronoi diagram, it could create arrangements where non-overlapping cells interact due to intersections of Voronoi faces (blue) with the bounding primitives. In addition to being unintuitive, this behavior would require the use of a full Delaunay adjacency graph in rendering, rather than the \ul{cheaper} \v{C}ech complex. The radical planes which define power faces (gold) can never create these non-local interactions.}
    \label{fig:nonlocal_faces}
  \end{minipage}
  \hfill
  \begin{minipage}[t]{0.48\linewidth}
    \begin{overpic}[width=\linewidth]{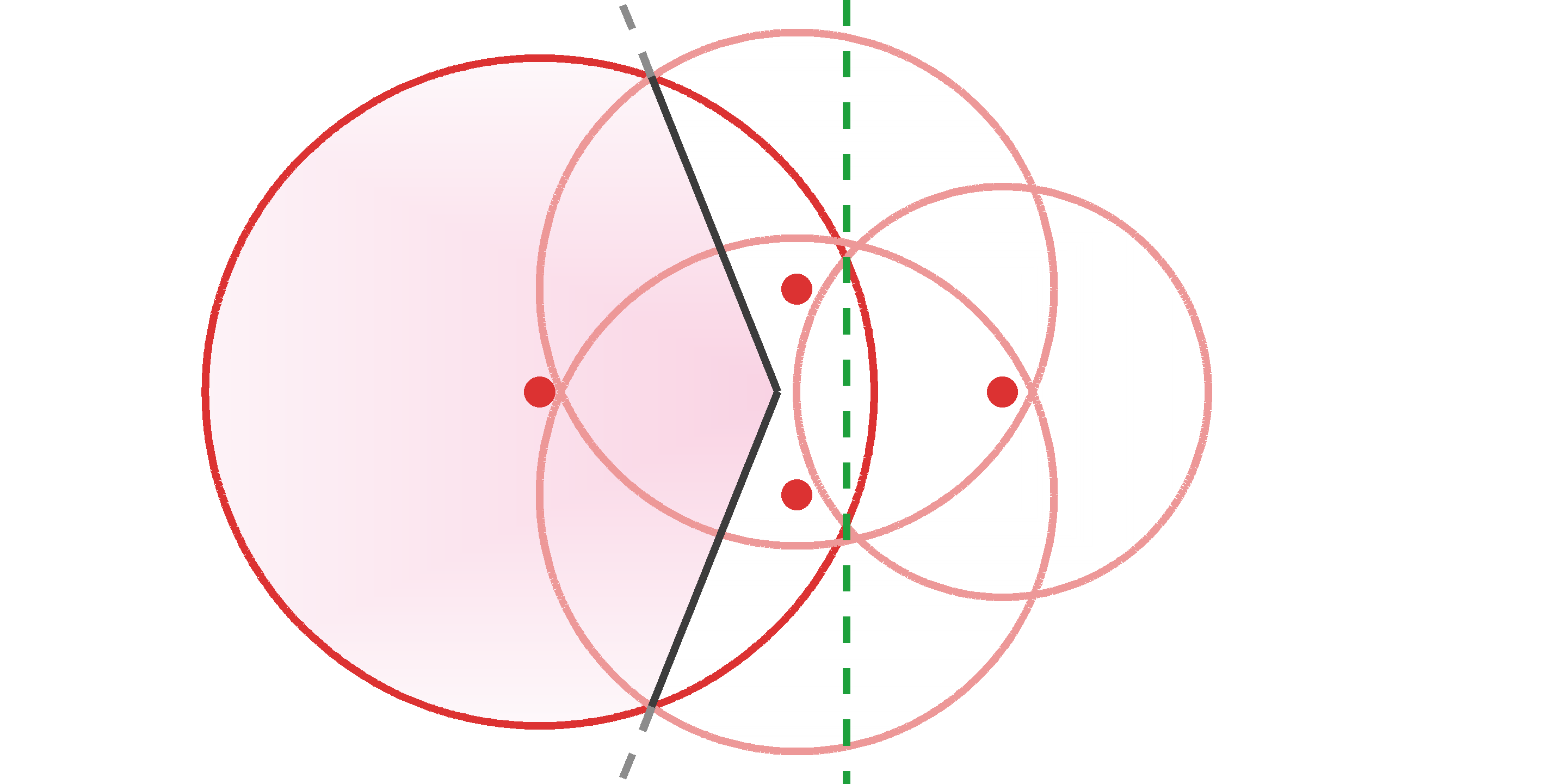}
    \end{overpic}
    \captionof{figure}{\textbf{Equivalence of rendering with the $\alpha$-complex and \v{C}ech complex} -- the dual graph of the bounded power diagram representation is the $\alpha$-complex, which is required during rendering to check for ray-face intersections. We can, however, avoid the cost of computing the $\alpha$-complex by taking advantage of the fact that its supersets, including the \v{C}ech complex, add spurious radical planes (green) that always lie entirely outside the cell, and thus have no effect on the result of rendering.}
    \label{fig:cech_equivalence}
  \end{minipage}
\end{figure}
Similarly to Radiant Foam, we can associate each bounded power cell with a density and directional radiance and compute pixel colors using the volume rendering equation; the real task of the rendering algorithm at this point is to enumerate the ray-cell intersections.
For ray tracing, the cell-to-cell traversal strategy of Radiant Foam still applies, and requires only that the Delaunay adjacency graph be replaced with the dual graph of the power diagram, in addition to considering the sphere bounds in the computation of intersection lengths.
We also observe improvement in ray tracing performance with the addition of Steiner points\footnote{Steiner points are additional points inserted into the triangulation to improve its quality; in our case to reduce the average number of neighbors for any given cell.}~\cite{courant1996mathematics}; see the \SupplementaryMaterial).

Rasterization of bounded power cells replaces per-ray traversal with a global sort by depths of the cell sites, similar to 3DGS.
However, unlike 3DGS, where the heuristic depth-sorting of semi-transparent splats inherently introduces view-dependent instability, our choice of power diagram as a parameterization of cells guarantees they are arranged such that \ul{no popping artifacts occur}.
This property was proven for Voronoi cells by Rebain et al.~\cite{derf}, and we extend this proof to power diagrams in the \SupplementaryMaterial.
In addition to completely avoiding a failure mode of splatting renderers that numerous papers~\cite{stopthepop, hahlbohm2025efficient, stochasticsplats, aaa, ever, condor2025don} have been dedicated to solving, this feature of our representation also enables lossless rasterization with non-pinhole cameras, such as the fisheye camera shown in~\cref{fig:teaser}; see the \Webpage.

\begin{figure}[t]
  \centering
  \begin{overpic}[width=\linewidth]{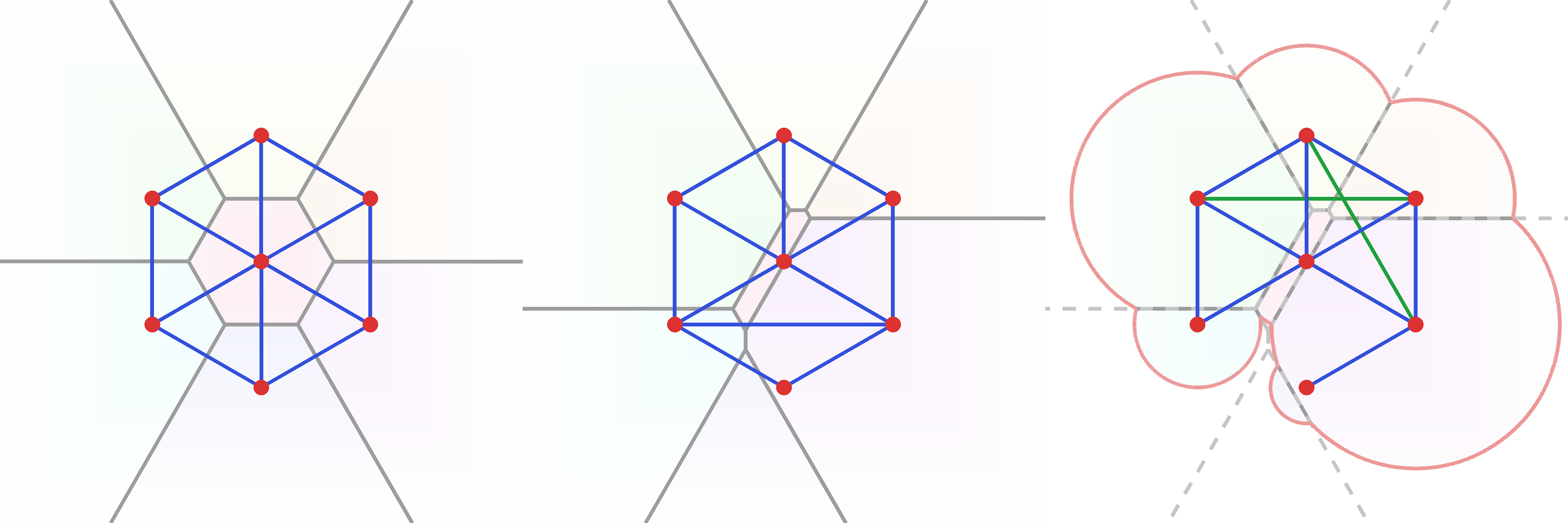}
  \end{overpic}
  \caption{\textbf{Comparison of adjacency graphs} -- Radiant Foam relied on computing the Delaunay triangulation (left) to provide the adjacency graph of its cells. While an unbounded power diagram would require a similar computation of a regular triangulation (center), the bounded power diagram requires only the $\alpha$-complex (right, blue), which excludes edges corresponding to non-overlapping spheres. We can also save computation by instead building the \v{C}ech complex -- the graph of all overlapping spheres -- which is a superset (right, blue+green) of the $\alpha$-complex. This approximation slows rendering slightly, but has \ul{no effect on the correctness} of the rendering; see \cref{fig:cech_equivalence}.}
  \label{fig:primal_meshes}
\end{figure}

The beneficial properties of bounded power cells do, however, come with a cost: we must iterate over the neighbors of each cell during rasterization to compute accurate intersection lengths between cell-to-cell faces.
We could again use the power diagram adjacency graph for this, though doing so during training would require constantly rebuilding a large regular triangulation over all cell sites which is very costly to compute repeatedly during training.
Instead, we rely on a simple geometric guarantee: if the bounding sphere of two cells do not overlap, it is impossible for them to share a face.
Consequently, we can safely exclude graph edges between them.
The subset of remaining valid graph edges -- the $\alpha$-complex of the sphere bounds -- is the minimal graph we can use for rasterization, yet extracting it dynamically remains expensive.

In practice, we find the graph of all overlapping spheres -- the Čech complex -- a better choice, as it contains all edges from the $\alpha$-complex, and is significantly cheaper to construct using GPU-accelerated collision detection, with the only cost being the introduction of some extraneous edges; see~\cref{fig:primal_meshes}.
Crucially, evaluating these false edges during rasterization preserves the exactness of the volume rendering integral, as the extraneous faces induced by them never intersect the true cell and thus don't change the intersection length; see~\cref{fig:cech_equivalence}.
This slight computational overhead of testing false edges slows down rendering during training by only around 10\%, which is vastly outweighed by the acceleration in graph construction.

\subsection{Oriented points representation}
During optimization, we observe that Radiant Foam naturally converges toward a bimodal density distribution, where cells exhibit either high density in occupied regions or near-zero density in empty space. 
This behavior implies that the scene's surface geometry is essentially defined by the sharp transition between these two states. 
However, relying on the interfaces between adjacent cells to capture this geometry is computationally inefficient, as it necessitates the explicit placement of zero-density points and \textit{wastefully} parametrizes these empty cells with model view-dependent appearance.
To eliminate this redundancy, we shift the boundary representation from the interface between cells to a localized interface within each cell. 
We achieve this by introducing an oriented point parameterization. 
This approach is analogous to a physical dipole -- where coupled ``charges'' of high and low density define a localized field -- providing a surface-aligned primitive that can represent both occupied and void space within a single volumetric cell.

Technically, we modify the cell parameterization from a single primal vertex $\primal_i$ to an oriented point defined by a face center $\primal_i$ and a normal vector $\normal_i$. 
This defines an internal ``oriented face'' that bisects the cell into two sub-regions. 
The~``inside'' half-space is assigned a learnable density and radiance, while the~``outside'' half-space is explicitly fixed to zero density. 
This surface-aligned representation reduces parameter redundancy by eliminating radiance parameters in empty space, and does not require changes to the rendering formulation, as it is essentially just the limiting case of two cells with sites that approach zero separation.

\subsection{Disentangling Geometry and Appearance model} 
\begin{figure*}[t]
    \centering
    \begin{overpic}[width=0.4\linewidth]{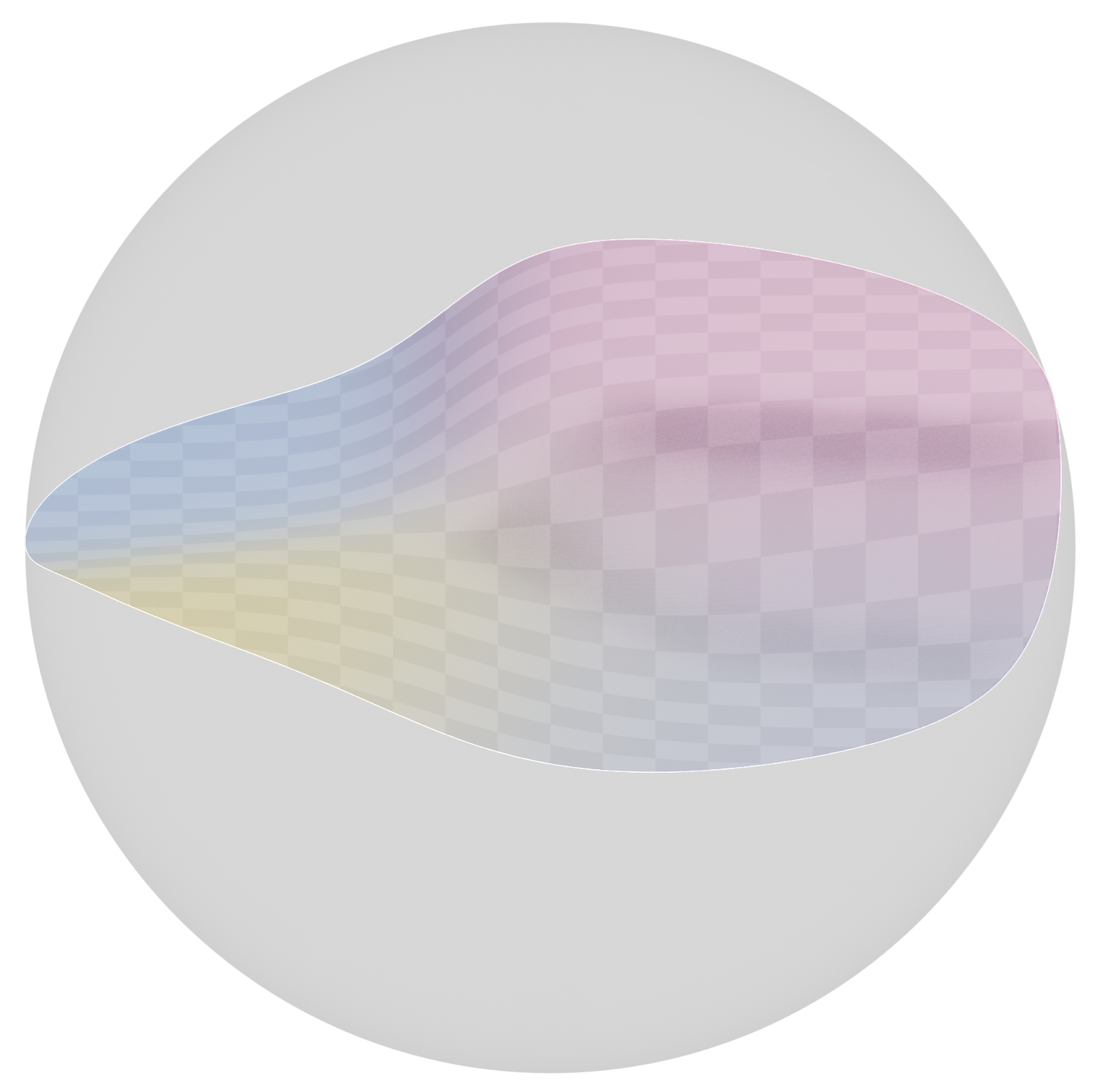}
    \end{overpic}
    \begin{overpic}[width=0.4\linewidth]{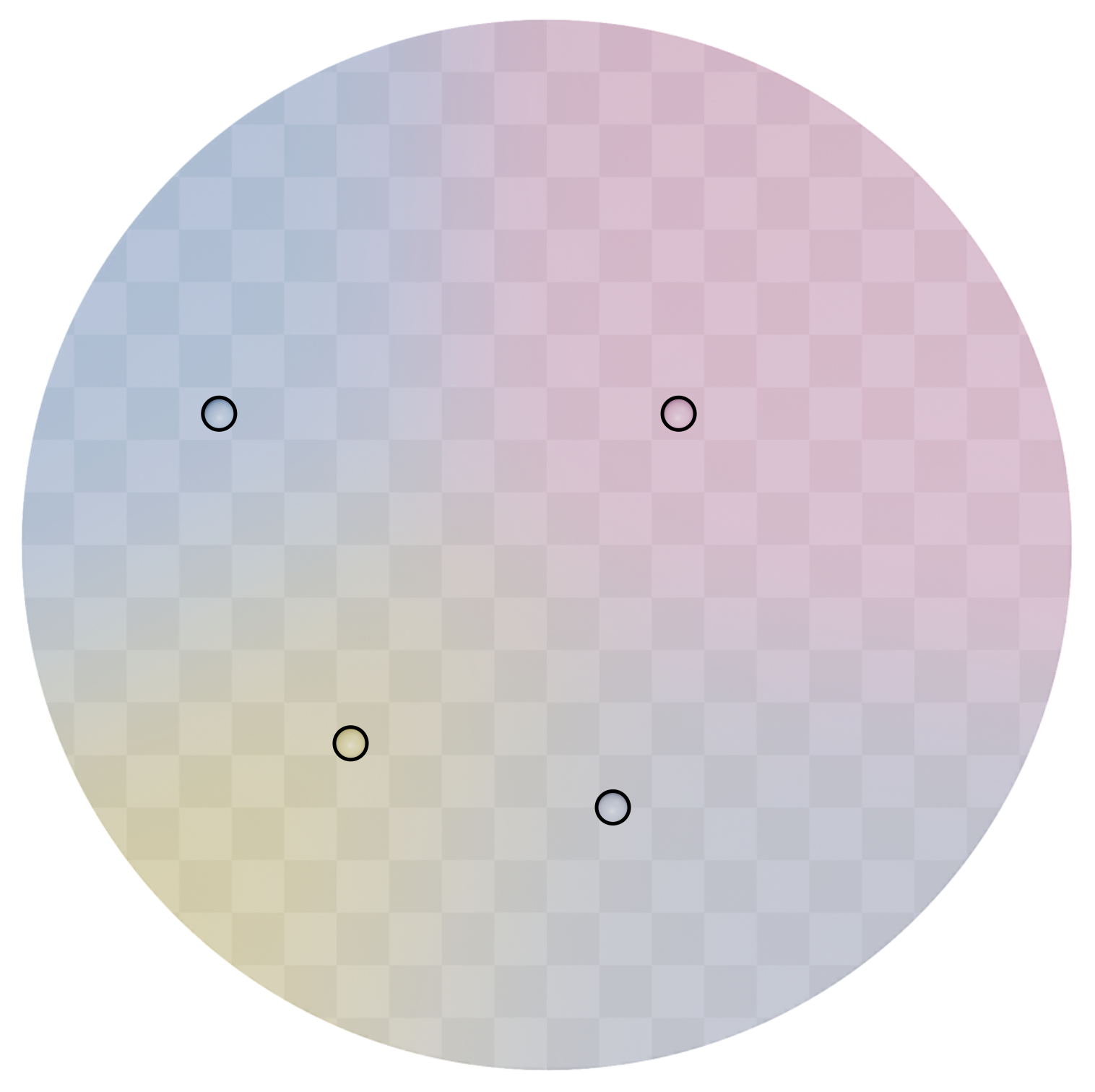}
    \put(12, 24){\scriptsize{($\mathbf{s}_1$, $\displacement_1$, $\radiance_1$)}}
    \put(5, 53){\scriptsize{($\mathbf{s}_2$, $\displacement_2$, $\radiance_2$)}}
    \put(48, 53){\scriptsize{($\mathbf{s}_3$, $\displacement_3$, $\radiance_3$)}}
    \put(45, 17){\scriptsize{($\mathbf{s}_4$, $\displacement_4$, $\radiance_4$)}}
    \end{overpic}
    \caption{
    {\bf Geometry and Appearance model -- }
    in our decoupled geometry and appearance framework, the dipole face acts as a proxy for macro-scale geometry, while detail sites $\mathbf{s}_i$ are optimized to capture high-frequency geometric and appearance details without increasing primitive count. 
    Displacement values $\displacement_i$ associated with each detail site push the surface up or down locally along the axis of the dipole (left).
    The soft Voronoi formulation (\cref{eq:svdisp} and \cref{eq:svrad}) distributes both the displacement and directional radiance $\radiance_i$ associated with each detail site across the dipole plane (right).
    }
    \label{fig:geo_app}
\end{figure*}
In computer graphics, it is a well-recognized principle that geometry and appearance often exhibit distinct frequency characteristics~\cite{texturedgaussians}.
Despite this, particle-based methods often couple geometry and outgoing radiance within a unified set of primitives. 
This coupling often results in primitives deployed to model high-frequency appearance details also redundantly modeling low-frequency macro-geometry, resulting in excessive parameter budgets. 
For example, while the coarse geometry of a textured wall can be efficiently represented by a small number of volumetric primitives, existing methods often utilize an excessive number of primitives with piecewise constant spherical functions to encode intricate textural details.

To address these problems, we propose to disentangle geometry and texture.
Our approach treats localized power cells as low-frequency geometric proxies, while high-frequency geometry and radiance are modeled as a continuous soft Voronoi function defined directly on the dipole faces; see ~\cref{fig:geo_app}.

Specifically, we associate each dipole face with a fixed number $k$ of learnable \textit{detail sites} $\mathbf{s}_i \in \mathbb{R}^2$. 
These sites serve a dual purpose: first, they act as anchors for a learnable displacement field $\displacement_i$, where defining offsets from the dipole faces at these sites allows us to effectively high-frequency geometric details without increasing the primary primitive count. 
Second, these same sites store directional radiance values $\radiance_i$ using Spherical Voronoi (SV) functions~\cite{spherical_voronoi}.

In the following, we describe the procedure for determining ray-dipole face intersections and computing the resulting color for this parameterization. 
To determine the cell radiance and intersection length, we first compute the initial intersection point $\bar{\mathbf{x}}$ between the ray and the base dipole face. 
The displacement at this intersection is then calculated using a soft Voronoi formulation:
\begin{align}
    \displacement(\bar{\mathbf{x}}) = \frac{\sum_i \exp(-\tau \|\bar{\mathbf{x}} - \mathbf{s}_i\|_2) d_i}{\sum_i \exp(-\tau \|\bar{\mathbf{x}} - \mathbf{s}_i\|_2)}
\label{eq:svdisp}
\end{align}
where $\tau$ denotes the temperature parameter controlling the smoothness of the soft Voronoi interpolation.
By displacing the dipole face along its normal direction according to this map, as illustrated in ~\cref{fig:geo_app}, we compute the final intersection point $\mathbf{x}$ between the ray and the displaced surface. 
This refined intersection point is utilized to compute the final intersection length and the surface radiance. 
The radiance at $\mathbf{x}$ is modeled using an analogous soft Voronoi interpolation:
\begin{align}
    \radiance(\mathbf{x}) = \frac{\sum_i \exp(-\tau \|\mathbf{x} - \mathbf{s}_i\|_2) \radiance_i}{\sum_i \exp(-\tau \|\mathbf{x} - \mathbf{s}_i\|_2)}
\label{eq:svrad}
\end{align}
This formulation allows a single geometric primitive to represent complex geometric and textural variations, significantly reducing the total primitive count required for high-fidelity reconstruction; see ~\cref{sec:train_details}.

\subsection{Optimization}
Following Radiant Foam~\cite{radfoam}, we recognize that the localized nature of our primitives renders the optimization landscape susceptible to local minima. 
To mitigate this, we adopt a two-pronged strategy: careful initialization followed by an adaptive schedule of \textit{densification} and \textit{pruning}.
Consistent with \cite{radfoam}, we initialize the optimization using a sparse point cloud generated via Structure-from-Motion~(SfM)~\cite{sfm}.

\paragraph{Densification and Pruning} 
We dynamically control the number of power cells to adaptively reallocate representational capacity. We manage this through two primary operations:
\begin{itemize}[label=\textbullet]
    \item \textbf{Densification:} We maintain an exponential moving average (EMA) of each primitive's photometric error. Candidates for densification are sampled from a multinomial distribution proportional to this error, targeting underfitting regions.
    \item \textbf{Pruning:} We track each primitive's accumulated contribution (defined as $\text{opacity} \times \text{transmittance}$) to the rendered pixels via a separate EMA. Primitives falling below a prescribed contribution threshold are pruned to remove redundant components.
\end{itemize}

\paragraph{Training objective}
We supervise the optimization process using a composite loss function:
\begin{equation}
    \mathcal{L} = \losst{rgb} + \lambda_{1}\losst{SSIM} + \lambda_{2}\losst{normal} + \lambda_{3}\losst{sparse} + \lambda_{4}\losst{connect}
\end{equation}
Beyond the standard L2 photometric and SSIM losses used to preserve perceptual detail, we employ three regularization strategies to prevent degenerate configurations (see the \SupplementaryMaterial for definitions of each term):
\begin{itemize}[label=\textbullet]
    \item \textbf{Normal loss ($\losst{normal}$):} Penalizes orientations where the dot product between the associated normal $\normal_i$ and the ray direction $\mathbf{d}$ is non-negative. This ensures the high-density dipole faces the camera, aligning primitives accurately with scene geometry.
    \item \textbf{Sparsity loss ($\losst{sparse}$):} Mitigates ``floating'' artifacts by applying an $L_1$ penalty to each primitive's contribution (defined as opacity $\times$ transmittance). This suppresses redundant primitives before they are pruned.
    \item \textbf{Connectivity loss ($\losst{connect}$):} Minimizes the radial overlap between adjacent spheres based on the \v{C}ech adjacencny graph. This maintains a sparse adjacency graph and ensures continuous surface geometry without excessive spatial redundancy.
\end{itemize}

\subsection{Implementation details}

\paragraph{Model Configurations}
Our main experiments use 8 detail sites per primitive (see~\cref{tab:ablation_dipoles} for ablation).
Each site is parameterized by a single Spherical Voronoi function, defined by eight spherical axes, as proposed by Di Sario et al.~\cite{spherical_voronoi}, and a scalar displacement value applied normal to the primitive's plane.

\paragraph{Training Details} 
\label{sec:train_details}
The model is optimized via a rasterization-based pipeline during the training phase. 
We utilize $500,000$ power sites for \dlthreedv and \mipnerf indoor scenes and $1.2$ million power sites for \mipnerf outdoor scenes -- $2\times$ to $4\times$ lower than existing baseline methods for the same datasets.
The training schedule begins with an initial $500$ iterations conducted on downsampled images to stabilize the global structure, followed by training at full resolution up to $30,000$ iterations. 
To adaptively refine the representation, we perform densification every $100$ iterations, beginning at iteration $1,000$ and concluding at iteration $24,000$. 
On an NVIDIA RTX 4090 GPU, our method requires $30$ minutes to train on the Bonsai scene from MipNeRF 360.

\section{Experiments}

\begin{table*}[t]
\centering
\caption{
    {\bf Qualitative comparisons --} we compare our method's novel view reconstruction accuracy and rendering speed to a number of ray tracing and rasterization baselines.
    For \mipnerf, we use the configuration provided by the authors for each method.
    For \dlthreedv, which serves as our test set and which none of the methods here provide configurations for, we use the corresponding configuration for \mipnerf indoor scenes.
    \ul{Important notes}: 3DGS, 3DGRT, and 3DGUT do not constrain the number of Gaussians in the configuration, but rather determine it \textit{dynamically} through optimization.
    Also, for some scenes the provided ray tracing implementation for Radiance Meshes failed, and for all others it was slower than our method.
    Per-scene breakdowns are provided in the~\SupplementaryMaterial.
}
\resizebox{\linewidth}{!}{
\setlength{\tabcolsep}{5pt}
\begin{tabular}{l|ccccc|ccccc}
    & \multicolumn{5}{c|}{\dlthreedv} & \multicolumn{5}{c}{\mipnerf} \\
    & PSNR$\uparrow$ & SSIM$\uparrow$ & LPIPS$\downarrow$ & \multicolumn{2}{c|}{FPS$\uparrow$} & PSNR$\uparrow$ & SSIM$\uparrow$ & LPIPS$\downarrow$ & \multicolumn{2}{c}{FPS$\uparrow$} \\
    & & & & ray & raster & & & & ray & raster \\
    \midrule 
    \multicolumn{1}{c}{} & \multicolumn{10}{l}{Rasterization only} \\
    \midrule 
    3DGS~\cite{3dgs} & 27.26 & 0.87 & 0.21 & - & \textbf{168} & 28.98 & 0.87 & 0.22 & - & 293 \\
    3DGS-MCMC~\cite{3dgs-mcmc} & 27.77 & 0.88 & 0.19 & - & 147 & 29.55 & \textbf{0.89} & \textbf{0.20} & - & \textbf{302} \\
    $\beta$-splat~\cite{betasplatting} & \textbf{28.55} & \textbf{0.89} & \textbf{0.18} & - & 139 &  \textbf{29.81} & \textbf{0.89} & \textbf{0.20} & - & 137\\
    \midrule
    \multicolumn{1}{c}{} & \multicolumn{10}{l}{Ray Tracing only} \\
    \midrule 
    3DGRT~\cite{3dgrt} & 27.28 & \textbf{0.87} & 0.25 & 82 & -  & 28.45 & \textbf{0.86} & \textbf{0.23} & 64 & - \\
    Radiant Foam~\cite{radfoam} & \textbf{27.80} & 0.86 & \textbf{0.23} & \textbf{112} & - & \textbf{28.47} & 0.83 & 0.24 & \textbf{180} & - \\
    \midrule
    \multicolumn{1}{c}{} & \multicolumn{10}{l}{Both} \\
    \midrule \\
    3DGUT~\cite{3dgut} & 27.74 & \textbf{0.87} & 0.26 & 63 & \textbf{208} & \textbf{28.98} & \textbf{0.87} & 0.23 & 55 & 177\\
    Radiance Meshes~\cite{radiancemeshes} & 26.05 & 0.84 & 0.30 & - & 174 & 28.39 & 0.86 & 0.26 & - & 159 \\
    Power Foam (ours) & \textbf{28.20} & 0.85 & \textbf{0.22} & \textbf{107} & 196 & 28.91 & 0.84 & \textbf{0.21} & \textbf{174} & \textbf{275} \\
\end{tabular}
}
\label{tab:qual_res}
\end{table*}

MipNeRF 360~\cite{mipnerf360} is the standard benchmark for this domain, but its lack of a private test set often leads to hyperparameter overfitting by baseline methods.
To ensure a rigorous evaluation, we follow standard tuning practices for MipNeRF 360 but also introduce the DL3DV sample set~\cite{dl3dv} as an untuned test set.
While DL3DV is a validated novel view synthesis dataset, none of our baselines report metrics on it, ensuring a fair evaluation of generalization.
{It is important to note that we \ul{use DL3DV as a true test set} -- we never trained a DL3DV scene while developing our method, and only after finalizing our method and hyperparameters did we evaluate on it \textit{once}}, using the same settings we used for MipNeRF 360~(specifically, the config for the indoor scenes).
For all baselines, we also train them on DL3DV using the hyperparameters they provide for MipNeRF 360.
In total, we evaluate on 18 scenes across DL3DV and MipNeRF 360.

To ensure our results are comparable with established baselines, we follow standard preprocessing protocols throughout our evaluation. 
For the DL3DV dataset, all images are consistently downsampled by a factor of two, while for the Mip-NeRF 360 dataset, indoor scenes are downsampled by a factor of two and outdoor scenes by a factor of four. 
All performance benchmarks, including frame rates and rendering speeds, were measured on a consumer-grade NVIDIA RTX 4090 GPU, which demonstrates the practical applicability of our method on modern, accessible hardware.

\paragraph{Metrics}
We assess the performance of each method using three standard image quality metrics, namely Peak Signal-to-Noise Ratio (PSNR), Structural Similarity Index (SSIM), and Learned Perceptual Image Patch Similarity (LPIPS). 
In addition to these quantitative measures, the \Webpage include \ul{rendered video paths} for selected scenes. 
These videos showcase the stability of our method when rendering from viewpoints that differ significantly from the training set distribution, providing a more holistic view of the reconstruction quality.

\paragraph{Quantitative Results}
Our quantitative evaluation on DL3DV and Mip-NeRF 360 is summarized in \Cref{tab:qual_res}. 
Crucially, our approach outperforms recent unified rendering methods—such as 3DGUT~\cite{3dgut} and Radiance Meshes~\cite{radiancemeshes}—in visual quality, establishing a new standard for representations that natively support both rasterization and ray tracing.

Furthermore, our method demonstrates exceptional efficiency across both rendering paradigms.
By constructing a full triangulation, our method can support constant-time single-ray traversal similar to Radiant Foam, which outperforms traditional BVH-based ray tracing methods.
Simultaneously, the localized characteristics of our method facilitate efficient rasterization, matching the rendering speed of 3DGS-based methods.
Finally, while achieving this dual-purpose efficiency, our method requires minimal compromise in absolute fidelity. 
Our visual quality remains highly competitive with pure rasterization models, performing comparably to the current differentiable rasterization state-of-the-art, such as 3DGS-MCMC~\cite{3dgs-mcmc} and $\beta$-splatting~\cite{betasplatting}.

\paragraph{Qualitative Results}
The baseline methods achieve high reconstruction fidelity on these standard benchmarks. 
Consequently, there are no easily perceivable differences between the high-quality outputs of our method and those of the baseline models in static images. 
We provide qualitative video comparisons between our model and the baseline models in the~\Webpage.

\begin{table}[t]
\begin{minipage}[t]{0.48\linewidth}
\centering
\caption{
    {\bf Ablation study --}
    we evaluate the impact of various components in our method by systematically excluding them and analyzing the reconstruction quality (PSNR $\uparrow$) on the Car and Statue scenes from \dlthreedv and Bonsai and Garden scenes from \mipnerf.
}
\resizebox{\linewidth}{!}{
\setlength{\tabcolsep}{2pt}
\begin{tabular}{cccccc|cccc|c}
    \rotatebox{90}{Radii} & \rotatebox{90}{Dipoles} & \rotatebox{90}{Disp.} & \rotatebox{90}{$\losst{connect}$} & \rotatebox{90}{$\losst{sparse}$} & \multicolumn{1}{c}{\rotatebox{90}{$\losst{normal}$}} & Car & Statue & Bonsai & \multicolumn{1}{c}{Garden} & Mean \\
    \midrule 
    \xmark & \cmark & \cmark & \cmark & \cmark & \cmark & 24.15 & 24.49 & 16.10 & 19.22 & 20.99 \\
    \cmark & \xmark & \cmark & \cmark & \cmark & \cmark & 29.51 & 29.24 & 30.75 & 24.32 & 28.46 \\
    \cmark & \cmark & \xmark & \cmark & \cmark & \cmark & 30.68 & 31.74 & 33.04 & 26.82 & 30.57 \\
    \cmark & \cmark & \cmark & \xmark & \cmark & \cmark & 30.71 & 31.47 & 32.52 & 26.45 & 30.29 \\
    \cmark & \cmark & \cmark & \cmark & \xmark & \cmark & 30.79 & 32.36 & 33.25 & 26.84 & 30.81 \\
    \cmark & \cmark & \cmark & \cmark & \cmark & \xmark & 30.78 & \textbf{32.70} & 33.21 & 26.90 & 30.90 \\
    \cmark & \cmark & \cmark & \cmark & \cmark & \cmark & \textbf{30.83} & 32.66 & \textbf{33.32} & \textbf{27.10} & \textbf{30.98} \\
\end{tabular}
}
\label{tab:ablation}
\end{minipage}
\hfill
  \begin{minipage}[t]{0.48\linewidth}
    \centering
    \caption{
        {\bf Ablation study --}
        we investigate the influence of detail site density by varying the number of sites per power cell and assessing the resulting reconstruction performance (PSNR $\uparrow$) across the Car and Statue scenes from \dlthreedv and Bonsai and Garden scenes from \mipnerf.
    }
    \resizebox{\linewidth}{!}{
    \setlength{\tabcolsep}{2pt}
    \begin{tabular}{c|cccc|c}
        \# of sites & Car & Statue & Bonsai & Garden & Mean \\
        \midrule 
        1 & 29.82 & 30.54 & 31.39 & 25.81 & 29.39 \\
        2 & 30.22	& 31.56	& 32.19	& 26.35	& 30.08 \\
        4 & 30.56	& 32.26	& 32.85	& 26.72	& 30.60 \\
        8 & \textbf{30.83} & \textbf{32.66} & \textbf{33.32} & \textbf{27.10} & \textbf{30.98} \\
    \end{tabular}
    }
    \label{tab:ablation_dipoles}
  \end{minipage}
\end{table}

\subsection{Ablation study}
We conducted a series of ablation experiments to isolate and quantify the impact of our architectural choices and loss functions. 
Specifically, we evaluated: per-cell learnable power radii, dipole parameterization, the number of detail sites, displacement mapping, and our regularization terms ($\losst{connect}$,  $\losst{sparse}$, and $\losst{normal}$). 
The quantitative results are summarized in ~\cref{tab:ablation} and ~\cref{tab:ablation_dipoles}.

\paragraph{Per-cell Learnable Radii}
We assessed the importance of per-cell radii by replacing them with a single, global radius for all power sites. 
This configuration resulted in a substantial decrease in reconstruction quality, demonstrating that fixed localization across the scene is insufficient. 
This drop highlights that different cells require varying spatial extents based on their location in the scene to effectively model diverse frequency information.

\paragraph{Dipole Parameterization}
To evaluate the impact of our dipole representation, we compared our approach against a baseline that models cells with constant density. 
Removing the dipole structure makes the model much less efficient in representing sharp surface boundaries, leading to a marked decline in reconstruction quality across all datasets.

\paragraph{Number of Detail Sites}
We investigated the scalability of appearance and geometry modeling by varying the number of detail sites per dipole plane (using 1, 2, 4, and 8 sites). 
As indicated in ~\cref{tab:ablation_dipoles}, increasing the number of sites consistently improves the PSNR. 
This suggests that additional degrees of freedom allow the model to represent increasingly complex appearances.

\paragraph{Displacement Field}
We ablated the role of the displacement field by eliminating displacement of the dipole plane during rendering. 
This omission led to a degradation in visual quality and PSNR, confirming that displacement fields are effective for efficiently capturing scene geometry.

\paragraph{Regularization Terms}
Finally, we ablated the three regularization losses. 
Removing the connectivity loss ($\losst{connect}$) resulted in the most notable performance degradation among the three, underscoring its role in maintaining spatial coherence. 
The sparsity ($\losst{sparse}$) and normal ($\losst{normal}$) losses also provided further incremental improvements to the overall reconstruction fidelity.

\section{Conclusion}
We have introduced Power Foam, a novel 3D representation that enables a unified rendering paradigm for both real-time ray tracing and rasterization. 
At the core of our approach is a foam-based structure composed of bounded polyhedral cells, which facilitates efficient rasterization while maintaining the inherent ray tracing efficiency of an explicit volumetric mesh. 
Our method produces \ul{mathematically identical results} under both rendering paradigms, avoiding the popping artifacts and view-inconsistency of splatting methods.
Furthermore, it matches the performance of state-of-the-art methods in their respective domains -- specifically Radiant Foam for ray tracing and 3D Gaussian Splatting for rasterization, providing a practical path toward unified real-time differentiable rendering.

%
%
\newpage

\paragraph{Acknowledgments}
We extend our deepest gratitude to George Shramko for his exceptional support and enormous help with early benchmarking.
This work was supported in part by the Natural Sciences and Engineering Research Council of Canada (NSERC) Discovery Grant, NSERC Collaborative Research and Development Grant, Google DeepMind, Digital Research Alliance of Canada, the Advanced Research Computing at the University of British Columbia, Microsoft Azure, and the SFU Visual Computing Research Chair program.

\bibliographystyle{splncs04}
\bibliography{main}

\begin{thebibliography}{10}
\providecommand{\url}[1]{\texttt{#1}}
\providecommand{\urlprefix}{URL }
\providecommand{\doi}[1]{https://doi.org/#1}

\bibitem{mipnerf}
Barron, J.T., Mildenhall, B., Tancik, M., Hedman, P., Martin-Brualla, R., Srinivasan, P.P.: Mip-nerf: A multiscale representation for anti-aliasing neural radiance fields. Int. Conf. Comput. Vis.  (2021)

\bibitem{mipnerf360}
Barron, J.T., Mildenhall, B., Verbin, D., Srinivasan, P.P., Hedman, P.: Mip-nerf 360: Unbounded anti-aliased neural radiance fields. IEEE Conf. Comput. Vis. Pattern Recog.  (2022)

\bibitem{texturedgaussians}
Chao, B., Tseng, H.Y., Porzi, L., Gao, C., Li, T., Li, Q., Saraf, A., Huang, J.B., Kopf, J., Wetzstein, G., Kim, C.: Textured gaussians for enhanced 3d scene appearance modeling. In: IEEE Conf. Comput. Vis. Pattern Recog. (2025)

\bibitem{condor2025don}
Condor, J., Speierer, S., Bode, L., Bozic, A., Green, S., Didyk, P., Jarabo, A.: Don't splat your gaussians: Volumetric ray-traced primitives for modeling and rendering scattering and emissive media. ACM Trans. Graph.  \textbf{44}(1),  1--17 (2025)

\bibitem{courant1996mathematics}
Courant, R., Robbins, H.: What is Mathematics?: an elementary approach to ideas and methods. OUP Us (1996)

\bibitem{spherical_voronoi}
Di~Sario, F., Rebain, D., Verbin, D., Grangetto, M., Tagliasacchi, A.: Spherical voronoi: Directional appearance as a differentiable partition of the sphere. arXiv preprint arXiv:2512.14180  (2025)

\bibitem{edelsbrunner2003shape}
Edelsbrunner, H., Kirkpatrick, D., Seidel, R.: On the shape of a set of points in the plane. IEEE Transactions on information theory  \textbf{29}(4),  551--559 (2003)

\bibitem{plenoxels}
Fridovich-Keil, S., Yu, A., Tancik, M., Chen, Q., Recht, B., Kanazawa, A.: Plenoxels: Radiance fields without neural networks. In: IEEE Conf. Comput. Vis. Pattern Recog. (2022)

\bibitem{radfoam}
Govindarajan, S., Rebain, D., Yi, K.M., Tagliasacchi, A.: Radiant foam: Real-time differentiable ray tracing. In: Int. Conf. Comput. Vis. pp. 4135--4145 (October 2025)

\bibitem{hahlbohm2025efficient}
Hahlbohm, F., Friederichs, F., Weyrich, T., Franke, L., Kappel, M., Castillo, S., Stamminger, M., Eisemann, M., Magnor, M.: Efficient perspective-correct 3d gaussian splatting using hybrid transparency. Comput. Graph. Forum  \textbf{44}(2),  e70014 (2025)

\bibitem{3dgs}
Kerbl, B., Kopanas, G., Leimk{\"u}hler, T., Drettakis, G.: 3d gaussian splatting for real-time radiance field rendering. ACM Trans. Graph.  \textbf{42}(4) (July 2023)

\bibitem{3dgs-mcmc}
Kheradmand, S., Rebain, D., Sharma, G., Sun, W., Tseng, Y.C., Isack, H., Kar, A., Tagliasacchi, A., Yi, K.M.: 3d gaussian splatting as markov chain monte carlo. In: Adv. Neural Inform. Process. Syst. (2024), spotlight Presentation

\bibitem{stochasticsplats}
Kheradmand, S., Vicini, D., Kopanas, G., Lagun, D., Yi, K.M., Matthews, M., Tagliasacchi, A.: Stochasticsplats: Stochastic rasterization for sorting-free 3d gaussian splatting. In: Int. Conf. Comput. Vis. pp. 26326--26335 (2025)

\bibitem{dl3dv}
Ling, L., Sheng, Y., Tu, Z., Zhao, W., Xin, C., Wan, K., Yu, L., Guo, Q., Yu, Z., Lu, Y., et~al.: Dl3dv-10k: A large-scale scene dataset for deep learning-based 3d vision. In: IEEE Conf. Comput. Vis. Pattern Recog. pp. 22160--22169 (2024)

\bibitem{betasplatting}
Liu, R., Sun, D., Chen, M., Wang, Y., Feng, A.: Deformable beta splatting. In: Proc. SIGGRAPH (2025)

\bibitem{neuralvolumes}
Lombardi, S., Simon, T., Saragih, J., Schwartz, G., Lehrmann, A., Sheikh, Y.: Neural volumes: Learning dynamic renderable volumes from images. ACM Trans. Graph.  \textbf{38}(4),  65:1--65:14 (Jul 2019)

\bibitem{ever}
Mai, A., Hedman, P., Kopanas, G., Verbin, D., Futschik, D., Xu, Q., Kuester, F., Barron, J., Zhang, Y.: Ever: Exact volumetric ellipsoid rendering for real-time view synthesis (2024), \url{https://arxiv.org/abs/2410.01804}

\bibitem{radiancemeshes}
Mai, A., Hedstrom, T., Kopanas, G., Kontkanen, J., Kuester, F., Barron, J.T.: Radiance meshes for volumetric reconstruction (2025), \url{https://arxiv.org/abs/2512.04076}

\bibitem{meijering1953interface}
Meijering, J.L.: Interface area, edge length, and number of vertices in crystal aggregates with random nucleation. Philips Research Reports  \textbf{8},  270--290 (1953)

\bibitem{nerf}
Mildenhall, B., Srinivasan, P.P., Tancik, M., Barron, J.T., Ramamoorthi, R., Ng, R.: Nerf: Representing scenes as neural radiance fields for view synthesis. In: Eur. Conf. Comput. Vis. (2020)

\bibitem{3dgrt}
Moenne-Loccoz, N., Mirzaei, A., Perel, O., de~Lutio, R., Esturo, J.M., State, G., Fidler, S., Sharp, N., Gojcic, Z.: 3d gaussian ray tracing: Fast tracing of particle scenes. ACM Trans. Graph.  (2024)

\bibitem{instantngp}
M\"uller, T., Evans, A., Schied, C., Keller, A.: Instant neural graphics primitives with a multiresolution hash encoding. ACM Trans. Graph.  \textbf{41}(4),  102:1--102:15 (Jul 2022)

\bibitem{stopthepop}
Radl, L., Steiner, M., Parger, M., Weinrauch, A., Kerbl, B., Steinberger, M.: Stopthepop: Sorted gaussian splatting for view-consistent real-time rendering. ACM Trans. Graph.  \textbf{43}(4),  1--17 (2024)

\bibitem{derf}
Rebain, D., Jiang, W., Yazdani, S., Li, K., Yi, K.M., Tagliasacchi, A.: Derf: Decomposed radiance fields. In: IEEE Conf. Comput. Vis. Pattern Recog. (2021)

\bibitem{sfm}
Sch\"{o}nberger, J.L., Frahm, J.M.: Structure-from-motion revisited. IEEE Conf. Comput. Vis. Pattern Recog.  (2016)

\bibitem{srn}
Sitzmann, V., Zollh{\"o}fer, M., Wetzstein, G.: Scene representation networks: Continuous 3d-structure-aware neural scene representations. In: Adv. Neural Inform. Process. Syst. (2019)

\bibitem{aaa}
Steiner, M., K{\"o}hler, T., Radl, L., Windisch, F., Schmalstieg, D., Steinberger, M.: Aaa-gaussians: Anti-aliased and artifact-free 3d gaussian rendering. In: Int. Conf. Comput. Vis. pp. 27650--27659 (2025)

\bibitem{digest}
Tagliasacchi, A., Mildenhall, B.: Volume rendering digest (for nerf). arXiv preprint arXiv:2209.02417  (2022)

\bibitem{3dgut}
Wu, Q., Martinez~Esturo, J., Mirzaei, A., Moenne-Loccoz, N., Gojcic, Z.: 3dgut: Enabling distorted cameras and secondary rays in gaussian splatting. IEEE Conf. Comput. Vis. Pattern Recog.  (2025)

\end{thebibliography}

\clearpage
\setcounter{page}{1}
\appendix

\section{Qualitative results}
We provide qualitative video comparisons in our \Webpage.

\section{Proof that the Power Diagram is Pop-free}
\label{sec:proof}

We first briefly restate the proof of Rebain et al.~\cite{derf}, which shows that Voronoi cells rendered in the depth order of their sites do not suffer from popping artifacts.
The original statement of the proof described this as compatibility with the Painter's Algorithm, which simply means that the ordering of ray-cell intersections for any ray is the same as the distance ordering of sites.

Let $\mathcal{P} = \{P_1, \dots, P_N\} \subset \mathbb{R}^n$ be a set of Voronoi sites. 
Recall that the \emph{Voronoi cell} of site $P$ is defined as:
\begin{equation}
  V_P \;=\; \bigl\{\, x \in \mathbb{R}^n \;\big|\;
    \|x - P\| \le \|x - P'\| \;\;\forall\, P' \in \mathcal{P} \bigr\}.
  \label{eq:voronoi_cell}
\end{equation}
Given a camera at position $Q \in \mathbb{R}^n$, define a partial order on the cells by:
\begin{equation}
  V_{P'} <_Q V_P
  \quad\Longleftrightarrow\quad
  d(P', Q) < d(P, Q),
  \label{eq:voronoi_order}
\end{equation}
where $d$ denotes Euclidean distance.

\begin{proposition}[\cite{derf}]
  \label{prop:voronoi_painter}
  The ordering~\eqref{eq:voronoi_order} is a valid painter's ordering: if any part of $V_{P'}$ occludes any part of $V_P$ as seen from~$Q$, then $V_{P'} <_Q V_P$.
\end{proposition}

\begin{proof}
  Suppose $x \in V_P$, $x' \in V_{P'}$, and $x'$ lies strictly between $x$ and $Q$ on the line segment joining them, i.e. $x' = \lambda\, x + (1-\lambda)\, Q$ for some $\lambda \in (0,1)$.
  We show that $d(P', Q) < d(P, Q)$.
  Define the \emph{halfspace}:
  \begin{equation}
    H \;=\; \bigl\{\, z \in \mathbb{R}^n \;\big|\;
      d(z, P) < d(z, P') \bigr\}.
    \label{eq:halfspace_voronoi}
  \end{equation}
  Since $x \in V_P$ we have $d(x,P) \le d(x,P')$, so $x \in \overline{H}$ (the closure of $H$).
  Since $x' \in V_{P'}$ we have $d(x',P') \le d(x',P)$, so $x' \notin H$.
  Because $H$ is a halfspace, its boundary $\partial H$ is a hyperplane, and any line can cross it at most once.
  If $Q$ were in $H$, the segment from $x$ (inside $\overline{H}$) to $Q$ (inside $H$) would have to re-enter $H$ after leaving it at $x'$, requiring two crossings of~$\partial H$---a contradiction.
  Therefore $Q \notin H$, which gives $d(Q, P) \ge d(Q, P')$.
  Generically the inequality is strict, yielding $V_{P'} <_Q V_P$.
\end{proof}

We can also state the proof in a more intuitive way: Voronoi faces are always orthogonal to the line segments connecting the corresponding sites.
Consequently, the orientation of any face in the mesh with respect to a camera origin is strictly determined by which site is closer to the camera -- thus a ray leaving one cell and entering another will always correspond to the distance of the site from the camera increasing.


\paragraph{Extension to power diagrams}
Recall that the power diagram generalizes the Voronoi diagram by assigning a real-valued weight $\omega_P=r_P^2 \in \mathbb{R}$ to each site $P \in \mathcal{P}$.
The \emph{power distance} from a point $z$ to site $P$ is defined as:
\begin{equation}
  \operatorname{pow}(z, P) \;=\; \|z - P\|^2 - \omega_P,
  \label{eq:pow_dist}
\end{equation}
and the \emph{power cell} of $P$ is:
\begin{equation}
  \Pi_P \;=\; \bigl\{\, x \in \mathbb{R}^n \;\big|\;
    \operatorname{pow}(x, P) \le \operatorname{pow}(x, P')
    \;\;\forall\, P' \in \mathcal{P} \bigr\}.
  \label{eq:power_cell}
\end{equation}
When all weights are equal, $\Pi_P = V_P$ and we recover the ordinary Voronoi diagram.

A key observation is that the bisecting surface between two power cells is the locus $\operatorname{pow}(z,P) = \operatorname{pow}(z,P')$, which expands to:
\begin{equation}
  2\,z \cdot (P' - P)
  \;=\;
  \|P'\|^2 - \|P\|^2 + \omega_{P'} - \omega_P.
  \label{eq:power_bisector}
\end{equation}
This is a hyperplane with normal $(P' - P)$, the \emph{same} normal as the ordinary Voronoi bisector; only the scalar offset changes.
In particular, every power cell is a convex polytope.

\medskip

We now prove that the painter's algorithm is compatible with power diagrams under a natural ordering based on power distance from the camera.

\begin{theorem}
  \label{thm:power_painter}
  Let $\mathcal{P} \subset \mathbb{R}^n$ be a set of sites with weights $\{\omega_P\}_{P \in \mathcal{P}}$, and let $Q \in \mathbb{R}^n$ be a camera position.
  Define the partial order:
  \begin{equation}
    \Pi_{P'} <_Q \Pi_P
    \quad\Longleftrightarrow\quad
    \operatorname{pow}(Q, P') \;<\; \operatorname{pow}(Q, P).
    \label{eq:power_order}
  \end{equation}
  Then this is a valid painter's ordering: if any part of\/ $\Pi_{P'}$ occludes any part of\/ $\Pi_P$ as seen from~$Q$, then $\Pi_{P'} <_Q \Pi_P$.
\end{theorem}

\begin{proof}
  Suppose $x \in \Pi_P$, $x' \in \Pi_{P'}$, and $x'$ lies on the open line segment between $x$ and $Q$, i.e.\ $x' = \lambda\, x + (1-\lambda)\, Q$ for some $\lambda \in (0,1)$.
  We show that $\operatorname{pow}(Q, P') < \operatorname{pow}(Q, P)$.
  Define the \emph{power halfspace}:
  \begin{equation}
    H \;=\; \bigl\{\, z \in \mathbb{R}^n \;\big|\;
      \operatorname{pow}(z, P) < \operatorname{pow}(z, P') \bigr\}.
    \label{eq:power_halfspace}
  \end{equation}
  Expanding,
  \begin{equation}
    H \;=\; \bigl\{\, z \;\big|\;
      2\,z \cdot (P' - P) \;<\;
      \|P'\|^2 - \|P\|^2 + \omega_{P'} - \omega_P
    \bigr\},
    \label{eq:power_halfspace_expanded}
  \end{equation}
  which is an open halfspace bounded by the hyperplane~\eqref{eq:power_bisector} with outward normal $(P' - P)$.
  The crucial property is that $H$ is still a halfspace -- the weights affect only the offset, not the orientation of the boundary.

  We now follow the same argument as in Proposition~\ref{prop:voronoi_painter}:
  \begin{enumerate}
    \item Since $x \in \Pi_P$, we have
      $\operatorname{pow}(x, P) \le \operatorname{pow}(x, P')$, so
      $x \in \overline{H}$.
    \item Since $x' \in \Pi_{P'}$, we have
      $\operatorname{pow}(x', P') \le \operatorname{pow}(x', P)$, so
      $x' \notin H$.
    \item The boundary $\partial H$ is a hyperplane, so any line crosses it
      at most once.
    \item Suppose for contradiction that $Q \in H$. Then the line segment from
      $x$ to $Q$ starts in $\overline{H}$ (at~$x$), exits $\overline{H}$ before
      reaching $x'$ (since $x' \notin \overline{H}$ or $x' \in \partial H$),
      and must re-enter $H$ to reach $Q$. This requires crossing $\partial H$
      at least twice---a contradiction, since $\partial H$ is a hyperplane.
    \item Therefore $Q \notin H$, which means
      $\operatorname{pow}(Q, P) \ge \operatorname{pow}(Q, P')$. For sites in
      general position the inequality is strict, giving
      $\operatorname{pow}(Q, P') < \operatorname{pow}(Q, P)$, i.e.\
      $\Pi_{P'} <_Q \Pi_P$.
  \end{enumerate}
\end{proof}

\section{Steiner points for ray tracing}

\begin{algorithm}[H]
\DontPrintSemicolon
\KwIn{Initial power cells $\powers = \{(\primal_i, \radius_i)\}_{i=1}^n$}

\caption{Steiner Point Insertion for Ray Tracing}
\label{alg:steiner_points}

\For{$iteration \leftarrow 1$ \KwTo $6$}{
    $\mathcal{S} \leftarrow \{\hat{\primal}_j \sim \mathcal{N}_S \subset \mathbb{R}^3 \}$\;
    
    \ForEach{$\hat{\primal}_j \in \mathcal{S}$}{
        \If{$\forall (\primal_i,\radius_i)\in\mathcal{P}:\;
            \|\hat{\primal}_j - \primal_i\|_2^2 - \radius_i^2 > 0$}{
            
            $(\primal_{near},\radius_{near})
            \leftarrow
            \argmin_{(\primal_i,\radius_i)\in\mathcal{P}}
            \left(
                \|\hat{\primal}_j - \primal_i\|_2^2 - \radius_i^2
            \right)$\;
            
            $d \leftarrow \|\hat{\primal}_j - \primal_{near}\|_2$\;
            $\hat{\radius}_j \leftarrow d - \radius_{near}$\;
            
            \If{$2 \radius_{near} \le \hat{\radius}_j \le 6 \radius_{near}$}{
                $\powers \leftarrow \powers \cup \{(\hat{\primal}_j,\hat{\radius}_j)\}$\;
            }
        }
    }
}
\end{algorithm}

To ray trace a bounded power diagram, we must construct the regular triangulation which is dual to the corresponding \textit{unbounded} power diagram, as we may need to traverse faces between power cells outside the sphere bounds.
However, because these unbounded parts of the cell do not affect rendering, they are effectively un-regularized during training, often resulting in thin and elongated cells. 
These suboptimal configurations force the ray to intersect an excessive number of power cells in empty regions, which significantly degrades ray tracing efficiency. 
To address this, we incorporate Steiner points -- a well-established concept in computer graphics -- to regularize the adjacency graph and enhance ray tracing performance.

We achieve this by progressively expanding the learned bounded power diagram, filling empty regions recursively with new cells. 
We sample random points within the 3D scene, specifically from a normal distribution $\mathcal{N}_S$ with the same mean and standard deviation as the scene points, discarding any that fall within existing power cells. 
For each valid candidate, we determine its nearest neighbor based on power distance and set the candidate's radius to the distance to that neighbor's sphere. 
We selectively retain new cells whose radius is between 2 to 6 times larger than that of the nearest neighbor.
This procedure is repeated over six recursive iterations to ensure the scene is filled with cells that facilitate a more uniform triangulation and a robust traversal structure, see ~\cref{alg:steiner_points}. 
Empirically, the introduction of these Steiner points reduces the average number of ray-cell intersections from 53.36 to 36.62 for the "Bonsai" scene, resulting in a performance gain from 113 to 185 FPS.

\section{Non-pinhole Rasterization}
There are two factors which contribute to the requirement of pinhole cameras for other methods like Gaussian splatting: first, linear approximations of the projection function which transforms 3D primitives into screen space break down for highly distorted camera models.
This can be addressed by either improving the approximation of projection, as in 3DGUT~\cite{3dgut}, or by adopting a primitive model which can be efficiently evaluated on a per-ray basis, such as in Radiance Meshes~\cite{radiancemeshes}.
Our method takes the second approach.

The second factor is reliance on rendering primitives in sorted order for correct occlusion in volume rendering.
Methods like 3DGS~\cite{3dgs} which are based on unstructured ``soups'' of primitives suffer from popping artifacts, where the correct traversal order of a ray through the primitives increasingly deviates from the depth order of those primitives for rays with a different direction than the central axis of the camera.
This has been addressed by approaches like per-ray sorting~\cite{stopthepop}, which while effective, adds cost and is incompatible with most hardware rasterization pipelines.
Alternatively, methods like ours, Radiant Foam~\cite{radfoam}, and Radiance Meshes~\cite{radiancemeshes} avoid this problem by employing mesh structures which admit an ordering of primitives which is correct for any ray that passes through the camera center, regardless of direction.

Qualitative examples demonstrating our support for lossless non-pinhole rasterization are available in the supplementary html page.

\section{Loss functions}
In addition to standard $L_2$ photometric and SSIM losses, we incorporate three auxiliary regularization terms during training. We describe these components in detail below:

\paragraph{Normal Loss}
This term ensures that surface normals are consistently oriented outward, preventing degenerate “back-facing” surface properties. We implement this by penalizing the positive dot product between the face normal $\normal_i$ and the ray direction $\raydirection_r$. For a given cell $\power_i$, the loss is formalized as:
\begin{align}
    \losst{normal}(\power_i) &= \sum_{\ray \in \trainrays} \transmittance_\ray \ \opacity_\ray \ \max(\normal_i \cdot \raydirection_\ray, 0)^2
\end{align}
where $\opacity_\ray$ denotes the opacity of the cell along the ray $\ray$, $\transmittance_\ray$ represents the transmittance of the primitive along the same ray and $\trainrays$ represents the set of all rays in the training set. 
This loss is initialized with a weight of $0.1$ and follows an exponential decay schedule to reach $0.01$ by the end of training across all datasets.

\paragraph{Sparsity Loss}
Inspired by the sparsity regularizer used in 3DGS-MCMC~\cite{3dgs-mcmc}, this term applies an $L_1$ penalty to the accumulated contribution of each primitive for a given training camera. 
This regularizer effectively suppresses "floaters" -- low-density artifacts in the scene -- which are subsequently removed during the pruning phase. 
The sparsity loss for cell $\power_i$ is defined as:
\begin{align}
    \losst{sparse}(\power_i) = \sum_{\ray \in \trainrays} \transmittance_\ray \ \opacity_\ray
\end{align}
where $\opacity_\ray$ denotes the opacity of the cell along the ray $\ray$, $\transmittance_\ray$ represents the transmittance of the primitive along the same ray and $\trainrays$ represents the set of all rays in the training set. 
We apply an initial weight of $0.1$, which is exponentially decayed to $0.0001$ by the end of training.

\paragraph{Connectivity Loss} 
This loss minimizes spatial overlap between adjacent cells to eliminate redundant connectivity, thereby facilitating efficient rasterization and a sparse adjacency graph. 
Specifically, we minimize the squared sum of the overlapping distances between a sphere and its neighbors in the \v{C}ech graph. 
For cell $\power_i$, this is expressed as:
\begin{align}
    \losst{connect}(\power_i) = \sum_{j \in \text{\v{C}ech}(i)} \max(\radius_i + \radius_j - d_{ij}, 0)^2
\end{align}
where $\radius_i$ and $\radius_j$ are the radii of cells $\power_i$ and $\power_j$, and $d_{ij}$ is the Euclidean distance between their primal vertices. 
We apply an initial weight of $1e-4$, which is exponentially decayed to $1e-7$ by the end of training.

\section{Per-scene quantitative comparisons}
~\cref{tab:dl3dv_psnr,tab:dl3dv_ssim,tab:dl3dv_lpips,tab:dl3dv_fps,tab:m360_psnr,tab:m360_ssim,tab:m360_lpips,tab:m360_fps} summarize the error metrics collected for our evaluation of all considered techniques. 
These include results for both \dlthreedv and \mipnerf scenes.

\begin{table}[h]
\centering
\caption{PSNR for \dlthreedv scenes}
\label{tab:benchmark_results}
\resizebox{\textwidth}{!}{
\setlength{\tabcolsep}{5pt}
\begin{tabular}{l|ccccccc|cccc}
    \toprule
    & \multicolumn{7}{c|}{\textbf{Indoor}} & \multicolumn{4}{c}{\textbf{Outdoor}} \\
    & roomset & herbary & vasary & supermarket & car & greenhouse & grills & garden & statue & 140 & highrise \\ 
    \midrule
    3DGS & 30.01 & 32.08 & 25.08 & 32.05 & 28.91 & 22.11 & 26.49 & 22.82 & 29.56 & 25.20 & 25.52 \\
    3DGS-MCMC & 30.02 & 32.19 & 26.06 & 31.82 & 28.91 & 22.27 & 27.24 & 23.34 & 31.01 & 25.53 & 27.11 \\
    $\beta$ splats & 30.73 & 32.81 & \textbf{26.43} & 32.54 & 29.95 & \textbf{22.94} & 28.25 & 23.56 & \textbf{32.79} & 26.00 & \textbf{28.07} \\
    3DGRT & 29.46 & 31.79 & 25.67 & 31.56 & 28.70 & 22.40 & 25.58 & 23.49 & 31.02 & 25.05 & 25.40 \\
    RadFoam & \textbf{31.21} & 31.86 & 24.74 & 32.05 & 29.75 & 22.00 & 26.65 & \textbf{24.54} & 31.49 & \textbf{26.70} & 24.77 \\
    Radiance Meshes & 25.17 & 30.86 & 23.72 & 28.70 & 28.47 & 22.13 & 23.34 & 23.64 & 30.62 & 24.88 & 25.05 \\
    3DGUT & 30.06 & 31.95 & 25.79 & 31.81 & 29.25 & 22.28 & 26.14 & 23.77 & 31.97 & 25.78 & 26.33 \\
    PowerFoam & 30.23 & \textbf{34.15} & 25.57 & \textbf{33.51} & \textbf{30.83} & 22.58 & \textbf{28.29} & 21.66 & 32.66 & 24.45 & 26.32 \\ 
    \bottomrule
\end{tabular}
}
\label{tab:dl3dv_psnr}
\end{table}

\begin{table}[h]
\centering
\caption{SSIM for \dlthreedv scenes}
\resizebox{\textwidth}{!}{
\setlength{\tabcolsep}{5pt}
\begin{tabular}{l|ccccccc|cccc}
    \toprule
    & \multicolumn{7}{c|}{\textbf{Indoor}} & \multicolumn{4}{c}{\textbf{Outdoor}} \\
    & roomset & herbary & vasary & supermarket & car & greenhouse & grills & garden & statue & 140 & highrise \\ 
    \midrule
    3DGS & 0.93 & 0.95 & 0.82 & 0.92 & 0.93 & 0.75 & 0.91 & 0.69 & 0.93 & 0.84 & 0.88 \\
    3DGS-MCMC & \textbf{0.94} & 0.95 & 0.84 & 0.92 & 0.93 & 0.76 & 0.92 & 0.71 & 0.94 & 0.85 & 0.91 \\
    $\beta$ splats & \textbf{0.94} & \textbf{0.96} & \textbf{0.85} & \textbf{0.93} & \textbf{0.94} & \textbf{0.79} & \textbf{0.93} & \textbf{0.72} & \textbf{0.95} &\textbf{ 0.86} & \textbf{0.91}\\
    3DGRT & 0.93 & 0.95 & 0.83 & 0.93 & 0.93 & 0.76 & 0.91 & 0.70 & 0.93 & 0.83 & 0.87 \\
    RadFoam & 0.93 & 0.94 & 0.78 & 0.92 & 0.92 & 0.72 & 0.88 & \textbf{0.72} & 0.93 & 0.84 & 0.84 \\
    Radiance Meshes & 0.84 & 0.94 & 0.77 & 0.90 & 0.93 & 0.75 & 0.82 & \textbf{0.72} & 0.94 & 0.82 & 0.85 \\
    3DGUT & 0.93 & 0.95 & 0.83 & 0.92 & 0.93 & 0.75 & 0.91 & 0.71 & 0.93 & 0.84 & 0.88 \\
    PowerFoam & 0.92 & \textbf{0.96} & 0.80 & \textbf{0.93} & 0.93 & 0.72 & 0.90 & 0.67 & 0.94 & 0.76 & 0.85 \\ 
    \bottomrule
\end{tabular}
}
\label{tab:dl3dv_ssim}
\end{table}

\begin{table}[h]
\centering
\caption{LPIPS for \dlthreedv scenes}
\resizebox{\textwidth}{!}{
\setlength{\tabcolsep}{5pt}
\begin{tabular}{l|ccccccc|cccc}
    \toprule
    & \multicolumn{7}{c|}{\textbf{Indoor}} & \multicolumn{4}{c}{\textbf{Outdoor}} \\
    & roomset & herbary & vasary & supermarket & car & greenhouse & grills & garden & statue & 140 & highrise \\ 
    \midrule
    3DGS & \textbf{0.10} & \textbf{0.11} & 0.23 & 0.19 & 0.20 & 0.24 & 0.24 & 0.28 & 0.27 & 0.25 & 0.17 \\
    3DGS-MCMC & \textbf{0.10} & \textbf{0.11} & 0.21 & 0.20 & 0.21 & 0.23 & 0.23 & 0.26 & 0.16 & 0.24 & 0.13 \\
    $\beta$ splats & \textbf{0.10} & \textbf{0.11} & \textbf{0.20} & 0.19 & 0.20 & \textbf{0.21} & \textbf{0.21} & \textbf{0.24} & \textbf{0.14} & \textbf{0.23} & \textbf{0.12} \\
    3DGRT & 0.16 & 0.19 & 0.25 & 0.24 & 0.29 & 0.27 & 0.30 & 0.30 & 0.25 & 0.31 & 0.22 \\
    RadFoam & 0.17 & \textbf{0.11} & 0.27 & 0.20 & 0.21 & 0.28 & 0.26 & 0.28 & 0.17 & 0.29 & 0.26 \\
    Radiance Meshes & 0.37 & 0.22 & 0.28 & 0.27 & 0.30 & 0.28 & 0.41 & 0.31 & 0.24 & 0.33 & 0.29 \\
    3DGUT & 0.17 & 0.19 & 0.28 & 0.25 & 0.29 & 0.29 & 0.31 & 0.29 & 0.25 & 0.31 & 0.22 \\
    PowerFoam & 0.17 & \textbf{0.11} & 0.27 & \textbf{0.16} & \textbf{0.19} & 0.28 & 0.24 & 0.31 & \textbf{0.14} & 0.34 & 0.21 \\ 
    \bottomrule
\end{tabular}
}
\label{tab:dl3dv_lpips}
\end{table}

\begin{table}[h]
\centering
\caption{Ray tracing / Rasterization FPS for \dlthreedv scenes}
\resizebox{\textwidth}{!}{
\setlength{\tabcolsep}{4pt}
\begin{tabular}{l|ccccccc|cccc}
    \toprule
    & \multicolumn{7}{c|}{\textbf{Indoor}} & \multicolumn{4}{c}{\textbf{Outdoor}} \\
    & roomset & herbary & vasary & supermarket & car & greenhouse & grills & garden & statue & 140 & highrise \\ 
    \midrule
    3DGS & -/152 & -/160 & -/98 & -/194 & -/195 & -/194 & -/151 & -/222 & -/220 & -/139 & -/118 \\
    3DGS-MCMC & -/133 & -/172 & -/95 & -/186 & -/163 & -/170 & -/135 & -/160 & -/180 & -/118 & -/109  \\
    $\beta$ splats & -/130 & -/158 & -/79 & -/184 & -/141 & -/165 & -/122 & -/134 & -/198 & -/117 & -/98 \\
    3DGRT & 88/- &86/- & 72/- & 80/- & 72/- & 102/- & 83/- & 83/- & 96/- & 78/- & 66/- \\
    RadFoam & \textbf{123}/- & \textbf{134}/- & \textbf{129}/- & 115/- & 107/- & \textbf{129}/- & 102/- & \textbf{129}/- & \textbf{109}/- & 72/- & 84/- \\
    Radiance Meshes & -/163 & 54/117 & -/\textbf{215} & -/194 & 50/194 & -/163 & -/175 & 22/147 & 90/176 & 32/172 & 9/\textbf{194} \\
    3DGUT & 61/233 & 56/\textbf{225} & 55/182 & 67/\textbf{221} & 53/233 & 110/\textbf{207} & 50/\textbf{227} & 76/\textbf{168} & 64/\textbf{229} & 55/186 & 48/178 \\
    PowerFoam & 113/\textbf{275} & 122/217 & 103/150 & \textbf{128}/185 & \textbf{120}/\textbf{267} & 98/183 & \textbf{105}/188 & 112/131 & 104/225 & \textbf{78}/\textbf{198} & \textbf{97}/136 \\ 
    \bottomrule
\end{tabular}
}
\label{tab:dl3dv_fps}
\end{table}
\begin{table}[h]
\centering
\caption{PSNR for \mipnerf scenes}
\label{tab:m360_psnr}
\resizebox{\textwidth}{!}{
\setlength{\tabcolsep}{5pt}
\begin{tabular}{l|cccc|ccc}
    \toprule
    & \multicolumn{4}{c|}{\textbf{Indoor}} & \multicolumn{3}{c}{\textbf{Outdoor}} \\
    & Room & Counter & Bonsai & Kitchen & Bicycle & Garden & Stump \\ 
    \midrule
    3DGS & 31.43 & 28.96 & 32.20 & 31.14 & 25.21 & 27.34 & 26.58 \\
    3DGS-MCMC & 32.05 & 29.32 & 32.66 & 31.91 & \textbf{25.69} & \textbf{27.81} & \textbf{27.38} \\
    $\beta$ splats & \textbf{32.64} & \textbf{30.14} & \textbf{33.79} & \textbf{32.04} & 25.52 & 27.62 & 26.95 \\
    3DGRT & 30.41 & 28.56 & 31.84 & 30.44 & 24.76 & 26.83 & 26.34 \\
    RadFoam & 30.87 & 28.58 & 32.22 & 31.28 & 24.23 & 26.56 & 25.53 \\
    Radiance Meshes & 30.83 & 28.27 & 31.36 & 30.71 & 24.75 & 26.32 & 26.49 \\
    3DGUT & 31.45 & 29.11 & 32.50 & 31.28 & 24.98 & 27.11 & 26.44 \\
    PowerFoam & 31.23 & 29.81 & 33.32 & 31.42 & 24.18 & 27.10 & 25.35 \\ 
    \bottomrule
\end{tabular}
}
\end{table}

\begin{table}[h]
\centering
\caption{SSIM for \mipnerf scenes}
\label{tab:m360_ssim}
\resizebox{\textwidth}{!}{
\setlength{\tabcolsep}{5pt}
\begin{tabular}{l|cccc|ccc}
    \toprule
    & \multicolumn{4}{c|}{\textbf{Indoor}} & \multicolumn{3}{c}{\textbf{Outdoor}} \\
    & Room & Counter & Bonsai & Kitchen & Bicycle & Garden & Stump \\ 
    \midrule
    3DGS & 0.91 & 0.91 & 0.94 & 0.92 & 0.77 & 0.87 & 0.78 \\
    3DGS-MCMC & \textbf{0.93} & \textbf{0.92} & \textbf{0.95} & \textbf{0.93} & \textbf{0.80} & \textbf{0.88} & \textbf{0.81} \\
    $\beta$ splats & \textbf{0.93} & \textbf{0.92} & \textbf{0.95} & \textbf{0.93} & 0.79 & 0.87 & 0.80 \\
    3DGRT & 0.91 & 0.91 & 0.94 & 0.92 & 0.75 & 0.85 & 0.77 \\
    RadFoam & 0.91 & 0.88 & 0.93 & 0.91 & 0.68 & 0.82 & 0.71 \\
    Radiance Meshes & 0.91 & 0.91 & 0.93 & 0.92 & 0.74 & 0.84 & 0.75 \\
    3DGUT & 0.92 & 0.91 & \textbf{0.95} & \textbf{0.93} & 0.76 & 0.85 & 0.77 \\
    PowerFoam & 0.91 & 0.90 & 0.94 & 0.92 & 0.68 & 0.83 & 0.71 \\ 
    \bottomrule
\end{tabular}
}
\end{table}

\begin{table}[h]
\centering
\caption{LPIPS for \mipnerf scenes}
\label{tab:m360_lpips}
\resizebox{\textwidth}{!}{
\setlength{\tabcolsep}{5pt}
\begin{tabular}{l|cccc|ccc}
    \toprule
    & \multicolumn{4}{c|}{\textbf{Indoor}} & \multicolumn{3}{c}{\textbf{Outdoor}} \\
    & Room & Counter & Bonsai & Kitchen & Bicycle & Garden & Stump \\ 
    \midrule
    3DGS & 0.29 & 0.26 & 0.25 & 0.16 & 0.24 & 0.12 & 0.25 \\
    3DGS-MCMC & 0.26 & 0.24 & 0.24 & 0.15 & \textbf{0.19} & \textbf{0.11} & \textbf{0.20} \\
    $\beta$ splats & 0.26 & 0.23 & 0.24 & 0.15 & 0.20 & 0.12 & 0.22 \\
    3DGRT & 0.30 & 0.26 & 0.25 & 0.17 & 0.26 & 0.15 & 0.26 \\
    RadFoam & 0.26 & 0.25 & 0.24 & 0.16 & 0.31 & 0.17 & 0.29 \\
    Radiance Meshes & 0.32 & 0.27 & 0.26 & 0.18 & 0.31 & 0.18 & 0.30 \\
    3DGUT & 0.30 & 0.26 & 0.25 & 0.16 & 0.24 & 0.14 & 0.26 \\
    PowerFoam & \textbf{0.22} & \textbf{0.20} & \textbf{0.19} & \textbf{0.14} & 0.31 & 0.14 & 0.30 \\ 
    \bottomrule
\end{tabular}
}
\end{table}

\begin{table}[h]
\centering
\caption{Ray tracing / Rasterization FPS for \mipnerf scenes}
\resizebox{\textwidth}{!}{
\setlength{\tabcolsep}{5pt}
\begin{tabular}{l|cccc|ccc}
    \toprule
    & \multicolumn{4}{c|}{\textbf{Indoor}} & \multicolumn{3}{c}{\textbf{Outdoor}} \\
    & Room & Counter & Bonsai & Kitchen & Bicycle & Garden & Stump \\ 
    \midrule
    3DGS & -/306 & -/302 & -/\textbf{472} & -/256 & -/\textbf{332} & -/169 & -/214 \\
    3DGS-MCMC & -/\textbf{467} & -/\textbf{355} & -/414 & -/\textbf{277} & -/184 & -/183 & -/\textbf{231} \\
    $\beta$ splats & -/178 & -/131 & -/182 & -/152 & -/88 & -/128 & -/100 \\
    3DGRT & 84/- & 57/- & 68/- & 33/- & 64/- & 71/- & 69/- \\
    RadFoam & \textbf{193}/- & \textbf{170}/- & \textbf{201}/- & 145/- & \textbf{173}/- & \textbf{216}/- & 164/- \\
    Radiance Meshes & -/196 & -/113 & 92/125 & -/212 & 100/205 & -/136 & 125/125 \\
    3DGUT & 58/253 & 44/243 & 51/239 & 29/171 & 57/97 & 68/127 & 78/109 \\
    PowerFoam & 179/393 & 167/271 & 185/353 & \textbf{150}/274 & 148/209 & 162/\textbf{244} & \textbf{234}/181 \\ 
    \bottomrule
\end{tabular}
}
\label{tab:m360_fps}
\end{table}

\end{document}